\documentclass[english,runningheads,11pt]{llncs}
\usepackage{tabularx,booktabs,multirow,delarray,array}
\usepackage{graphicx,amssymb,amsmath,amssymb}
\usepackage{latexsym}
\usepackage{epstopdf}

\usepackage[in]{fullpage}
\def\etal{\textsl{et~al. }}

\def\dist{\mathrm{d}}
\def\ANN{\mathop{\mathrm{ANN}}}
\def\AFN{\mathop{\mathrm{AFN}}}

\def\Ed{\mathsf{A}\dist}

\def\maxop{{\sc Max}}
\def\sumop{{\sc Sum}}

\def\annmax{ANN-\maxop}

\def\calC{\mathcal{C}}
\def\calA{\mathcal{A}}
\def\calD{\mathcal{D}}

\newenvironment{proof}{\noindent {\textbf{Proof:}}\rm}{\hfill $\Box$\rm}

\makeatletter

\newcommand{\mparagraph}[1]{\vspace{1.0ex \@plus1ex
      \@minus.2ex}\noindent\textbf{#1}\hspace{1em}}

\makeatother

  \makeatletter
  \long\def\@makecaption#1#2{
    \vskip 10pt
    \setbox\@tempboxa\hbox{{\footnotesize \textbf{#1.} #2}}
    \ifdim \wd\@tempboxa >\hsize         % IF longer than one line:
        {\footnotesize \textbf{#1.} #2\par}% THEN set as ordinary paragraph.
      \else                              %   ELSE  center.
        \hbox to\hsize{\hfil\box\@tempboxa\hfil}
    \fi}
 \makeatother
\begin{document}

\title{On Top-$k$ Weighted {\sc Sum} Aggregate Nearest and
	Farthest Neighbors in the $L_1$ Plane}

%\author{Haitao Wang and Wuzhou Zhang}
%\institute{Utah State University and Duke University}

\author{Haitao Wang\inst{1} \and Wuzhou Zhang\inst{2}}

 \institute{
    Department of Computer Science\\
   Utah State University, Logan, UT 84322, USA\\
   \email{haitao.wang@usu.edu}
   \and
   Department of Computer Science\\
   Duke University, Durham, NC 27708, USA\\
   \email{wuzhou@cs.duke.edu}
 }

\maketitle

\pagenumbering{arabic}
\setcounter{page}{1}

%\vspace*{-0.2in}
\begin{abstract}
In this paper, we study top-$k$ aggregate (or group) nearest neighbor queries using the weighted \sumop\
operator under the $L_1$ metric in the plane. Given a set $P$ of $n$ points, for any query consisting of a set
$Q$ of $m$ weighted points and an integer $k$, $ 1 \le k \le n$, the top-$k$ aggregate
nearest neighbor query asks for the $k$ points of $P$ whose
{\em aggregate distances} to $Q$ are the smallest, where the aggregate
distance of each point $p$ of $P$ to $Q$ is the sum of the
weighted distances from $p$ to all points of $Q$.
We build an $O(n\log n\log\log n)$-size data structure in $O(n\log n \log\log n)$
time, such that each top-$k$ query can be answered in
$O(m\log m+(k+m)\log^2 n)$ time. 
We also obtain other results with trade-off between preprocessing and query.
%With trade-off between preprocessing and query time, we also build two other data structures: the first one has $O(n\log n)$ preprocessing time and space with $O(m\log m+(k+m)\log^2 n\log\log n)$ query time; the second one has $O(n\log n\log^* n)$ preprocessing time and space with $O(m\log m+(k+m)\log^2 n\log^* n)$ query time.
Even for the special case where $k=1$, our results are better than
the previously best method (in PODS 2012), which
requires $O(n\log^2 n)$ preprocessing time,
$O(n\log^2 n)$ space, and $O(m^2\log^3 n)$ query time.
In addition, for the one-dimensional version of this problem, our
approach can build
an $O(n)$-size data structure in $O(n\log n)$ time that can support
$O(\min\{k,\log m\}\cdot m+k+\log n)$  time queries.
Further, we extend our techniques to the top-$k$ aggregate {\em
	farthest}
neighbor queries, with the same bounds.
\end{abstract}

\section{Introduction}
\label{sec:intro}
Top-$k$ nearest neighbor searching has been well-studied, e.g., see \cite{ref:AurenhammerVo00,ref:ClarksonNe06} for a survey.
For a set $P$ of points in the $d$-D space $\mathbb{R}^d$, the problem asks for a
data structure to quickly report the $k$ nearest neighbors in $P$ for any
query point. Aggregate nearest neighbor ($\ANN$) searching,
also known as {\em group} nearest neighbor searching \cite{ref:AgarwalNe12,ref:LiGr11,ref:LiTw05,ref:LiFl11,ref:LianPr08,ref:LuoEf07,ref:PapadiasGr04,ref:PapadiasAg05,ref:SharifzadehVo10,ref:YiuAg05},
is a generalization of the basic problem,
where each query consists of a set of (weighted)
points and the result of the query is based on applying
{\em aggregate} operators, such as (weighted) \sumop ~and \maxop, on all the points in the query.
In this paper, we study top-$k$ $\ANN$ queries using the weighted \sumop\ operator under the $L_1$ metric in the plane.

%In many applications, e.g. face recognition and sensor networks, data
%is inherently imprecise due to various reasons, such as noise or
%multiple observations. Numerous classic problems, including
%clustering~\cite{ref:CormodeAp08},
%skylines~\cite{ref:AfshaniAp11,ref:PeiPr07},
%range queries~\cite{ref:AgarwalIn09}, and nearest neighbor
%searching~\cite{ref:AgarwalNe12,ref:YuenSu10},
%have been cast and studied under uncertainty in the past few years.
%In this paper, we consider the top-$k$ nearest neighbor searching
%where the query data is uncertain. Further, we focus on the distances measured by
%the $L_1$ metric, which is appropriate for
%applications like VLSI design automation and urban transportation modeling (``Manhattan metric'').
% This problem has been studied by Agarwal {\em et al.} \cite{ref:AgarwalNe12} and we propose a better solution in this paper. The same problems with Euclidean distance measure and squared Euclidean distance measure were also studied in \cite{ref:AgarwalNe12}. The converse problem model where the input data are uncertain and the query data are certain was also considered in \cite{ref:AgarwalNe12}. Refer to \cite{ref:AgarwalNe12} for motivations of these problems.

\subsection{Problem Statement, Previous Work, and Our Results}
%\mparagraph{Problem statement, previous work, and our results.}
For any two points $p$ and $q$ in the plane, denote by
$\dist(p,q)$ the distance of $p$ and $q$.
Let $Q$ be a set of points and each point $q\in Q$ has a weight $w(q) > 0$.
Throughout the paper, we use $m$ to denote the size of $Q$ (note that
$m$ is not a fixed value). For any point $p$ in the plane,
the \textit{aggregate distance} from $p$ to $Q$,
denoted by $\Ed(p,Q)$, is defined to be
$$\Ed(p, Q) = \sum_{q\in Q} w(q) \dist(p, q).$$

%$\Ed(p, Q) = \displaystyle\sum_{i = 1}^m w_i \dist(p, q_i)$.
%$\Ed(p, Q) = \sum_{i = 1}^k w_i \dist(p, q_i)$.

Let $P$ be a set of $n$ points in
the plane. Given a query consisting of a set $Q$ of weighted points and an integer $k$, $1\leq k\leq n$, the \textit{top-$k$ aggregate nearest neighbors}
(top-$k$ $\ANN$s) of $Q$ in $P$ are the $k$ points of $P$ whose aggregate
distances to $Q$ are the smallest among all points in $P$; we denote
by $S_k(P,Q)$ the set of the top-$k$ $\ANN$s.
% (in particular, when $k=1$, $S_1(P,Q)$ is the $\ANN$ of $Q$ in $P$).
Our goal is to design a data structure to quickly report the
set $S_k(P,Q)$ for any query set $Q$ and $k$.

In this paper, we consider the $L_1$ metric. Specifically, for any two points $p = (x(p),y(p))$ and $q = (x(q),y(q))$ in the plane, their distance is defined to be
$\dist(p,q)=|x(p)-x(q)|+|y(p)-y(q)|$.
We build an $O(n\log n\log\log n)$-size data structure in $O(n\log n\log\log n)$ time
that can support each query
in $O(m\log m+(k+m)\log^2 n)$ time. Note that we also return the
aggregate distance of each point in $S_k(P,Q)$ to $Q$  and the points
		of $S_k(P,Q)$ are actually reported in sorted order by their
		aggregate distances to $Q$. With trade-off between preprocessing and query time, we also build two other data structures: the first one has $O(n\log n)$ preprocessing time and space with $O(m\log m+(k+m)\log^2 n\log\log n)$ query time; the second one has $O(n\log n\log^* n)$ preprocessing time and space with $O(m\log m+(k+m)\log^2 n\log^* n)$ query time.

For the 1-D version of this problem, our
approach can build an $O(n)$-size data structure in $O(n\log n)$ time with
$O(\min\{k,\log m\}\cdot m+k+\log n)$  query time, and the query time
can be reduced to $O(k+m+\log n)$ time if the points of $Q$ are given
in sorted order.
%Note that in the $1$-D space, the
%$L_1$ metric is the same as the $L_2$ metric.
%For the $L_2$ metric, only approximation results have been
%given in $\mathbb{R}^d$ when $d\geq 2$, e.g., \cite{ref:AgarwalNe12,ref:LiFl11}.

Further, we extend our techniques to solve the top-$k$ {\em aggregate
farthest neighbor} ($\AFN$) searching problem, with the same bounds as above.

\subsection{Related Work}
%\mparagraph{Related work.}

Previously, only approximation and heuristic results were given for
the top-$k$ ANN query problem \cite{ref:LjosaAP07}.
For the special case where $k=1$, Agarwal \etal\cite{ref:AgarwalNe12}
built an $O(n\log^2 n)$-size data structure in $O(n\log^2 n)$ time
that can answer each top-$1$ $\ANN$ query in $O(m^2\log^3 n)$ time.
Hence, even for the special case where $k=1$, our results are better
than that in \cite{ref:AgarwalNe12} in all three aspects:
preprocessing time, space, and query time. Recently, Ahn {\em et al.}
\cite{ref:AhnGr13} studied the {\em unweighted
version} of the problem, where they gave two data structures under
the assumption that the maximum value of $|Q|$ is known in advance
as $m$ for all queries, with the following time bounds: the first
one is built in $O(m^2 n \log^2 n)$ time and space with $O(m^2\log
		n +k(\log\log n+\log m))$ query time; the second one is
built in  $O(m^2 n\log n)$ time and $O(m^2 n)$ space with
$O(m^2\log n+(k+m)\log^2 n)$ query time. Clearly, our results,
albeit on the weighted version and do not require the assumption, are
generally better than the results in \cite{ref:AhnGr13} for most cases
(e.g., if $m=O(1)$, their second result is better than ours on the
 unweighted version with the assumption).

For the $L_2$ metric (i.e., the Euclidean distance),
only heuristic and approximation
algorithms were known previously for answering even top-$1$ $\ANN$
queries~\cite{ref:AgarwalNe12,ref:LiTw05,ref:LiFl11,ref:LuoEf07,ref:PapadiasGr04,ref:PapadiasAg05,ref:SharifzadehVo10,ref:YiuAg05};
the best known heuristic method for the top-1 $\ANN$ queries is based on R-tree~\cite{ref:PapadiasAg05}, and Li \etal\cite{ref:LiFl11} gave
a data structure with 3-approximation query performance
for the top-1 $\ANN$. Agarwal {\em et al.}
\cite{ref:AgarwalNe12} gave a data structure with a polynomial-time
approximation scheme for the top-1 $\ANN$ queries.

If the \maxop\ operator is used to define the aggregate distance, i.e.,
$\Ed(p, Q) = \max_{q\in Q} w(q) \dist(p, q)$,
	we refer the problem as top-$k$ \annmax\ queries.
To the best of our knowledge, we are not aware of any previous work on
the general weighted top-$k$ \annmax\ queries, even for $k=1$. Below is some previous work on the unweighted versions.
For top-1 \annmax\ queries, Papadias et al. \cite{ref:PapadiasAg05}
presented a heuristic Minimum Bounding Method with worst case query time
$O(n+m)$ for $L_2$ metric. Recently, Li et al. \cite{ref:LiGr11}
gave more results on the $L_2$ top-1
\annmax\ queries (the queries were called {\em group
enclosing queries}): by using $R$-tree
\cite{ref:GuttmanR84}, they \cite{ref:LiGr11}
 gave an exact algorithm that is very fast in practice although theoretically the
worst case query time is $O(n+m)$; they \cite{ref:LiGr11}
also gave a $\sqrt{2}$-approximation algorithm with query time
$O(m+\log n)$ for any fixed dimensions and
they further extended the algorithm to obtain a
$(1+\epsilon)$-approximation result.
Wang \cite{ref:WangAg13} gave an exact algorithm
that can answer each $L_2$ top-1 \annmax\ query in $O(m\sqrt{n}\log^{O(1)}n)$ time.
For the $L_1$ metric, Wang \cite{ref:WangAg13} constructed a data
structure of $O(n)$ size in $O(n\log n)$ time that can answer each
(unweighted) $L_1$ top-$k$ \annmax\ query in $O(m+k\log n)$ time.

In addition, Li et al. \cite{ref:LiFl11} proposed the {\em
flexible} top-$k$ $\ANN$ queries, which extend the classical $\ANN$ queries, and
they provided constant ratio approximation algorithms that work for
both \sumop\ and \maxop\ operators in any metric space and any fixed dimension.

We should point out that the weighted $\ANN$ queries studied in this
paper can be used to solve the {\em expected nearest neighbor} (ENN)
queries for uncertain query points under $L_1$ metric. In each ENN query, an
uncertain point $Q$ is given with $m$ different locations and each
location $q$ is associated with a probability $w(q)$ of being the
true location of $Q$, and the query asks for the point in $P$ that
has the smallest expected distance to $Q$. Agarwal {\em et al.}
\cite{ref:AgarwalNe12} gave the first nontrivial methods for answering exact or
approximate ENN queries under $L_1$, $L_2$, and the squared Euclidean distance,
with provable performance guarantees. We have mentioned their exact
top-1 query algorithm on $L_1$ metric earlier.
Other formulations on nearest neighbor queries over uncertain data
have also been studied in  \cite{ref:AgarwalNe12}
and elsewhere, e.g., \cite{ref:AgarwalNe13,ref:BeskalesEf08,ref:ChengPr08,ref:ChengUV10,ref:LjosaAP07,ref:YuenSu10}.

%In addition, Li and Deshpande \cite{ref:LiRa10} studied the
%top-$k$ nearest neighbor searching under the probabilistic model.

For the top-$k$ $\AFN$ queries, to the best of our knowledge, we are not
aware of any previous work on the weighted queries. For unweighted
queries, Gao {\em et al.} \cite{ref:GaoAg11} gave heuristic algorithms
using R-trees for the $L_2$ metric. For the $L_1$ metric, Ahn {\em et
	al.} \cite{ref:AhnGr13} also extended their techniques to top-$k$
	$\AFN$ queries with the same time bounds, assuming that the maximum
	value $m$ is known for all queries. For $k=1$, farthest Voronoi diagrams \cite{ref:AurenhammerVo00} can be used for answering top-1 $\AFN$ queries.

The rest of the paper is organized as follows. In Section
\ref{sec:1d}, we give our results in the 1-D space, which
are generalized to the 2-D space in Section \ref{sec:2d}.
One may view Section \ref{sec:1d} as a ``warm-up'' for Section
\ref{sec:2d}. Section \ref{sec:AFN} extends our techniques to solve
the $\AFN$ queries.
Section \ref{sec:conclusion} concludes the paper.

For simplicity of discussion, we make a general
position assumption that no two
points in $P\cup Q$ have the same $x$- or $y$-coordinate for any query
$Q$; we also assume no two points of $P$ have the same aggregate distance to
$Q$.  Our techniques can be extended to the general case without these assumptions, although the discussion would be more tedious.

Throughout the paper, we use $Q$ to denote the uncertain query
point and assume $k<n$.
To simplify the notation, we will write $\Ed(p)$ for
$\Ed(p,Q)$, and $S_k(P)$ for $S_k(P,Q)$.
When we say ``the $\ANN$'', we mean the top-1 $\ANN$.
For any subset $P'\subseteq P$, denote by $S_k(P')$ the set of the
top-$k$ $\ANN$s of $Q$ in $P'$.
Let $W=\sum_{q\in Q}w(q)$. We assume $m< n$ always holds since otherwise we could compute $S_k(P)$ in $O((m+n)\log m)=O(m\log m)$ time by directly computing the aggregate distances for all points in $P$, and we omit the details.
%Although $W=1$ in our definition, our method is applicable to the case $W\neq 1$ (e.g., the weighted \textsc{Sum} $\ANN$).

\section{Top-$k$ $\ANN$ Searching in the 1-D Space}
\label{sec:1d}
In the $1$-D space, all points of $P$ lie on a real line $L$. We assume $L$ is the $x$-axis.
For any point $p$ on $L$, denote by $x(p)$ the coordinate of $p$ on $L$.
Consider any query set $Q=\{q_1,\ldots, q_m\}$ on $L$.
For any point $p$ on $L$, the aggregate distance from $p$ to $Q$ is
$\Ed(p)=\sum_{q\in Q}w(q)\dist(p,q)$, where $\dist(p,q)=|x(p)-x(q)|$.
Given any $Q$ and any $k$, our goal is to compute $S_k(P)$, i.e., the set of the
 top-$k$ $\ANN$s of $Q$ in $P$.

For a fixed query set $Q$, a point $p$ on $L$ is called a {\em global minimum point} if it minimizes the
aggregate distance $\Ed(p)$ among all points on $L$. Such a global minimum point on $L$ may not be unique.
The global minimum point is also known as weighted Fermat-Weber point \cite{ref:DurierGe85}, and as shown below, it is very easy to compute in our problem setting.

To find $S_k(P)$, we will use the following strategy. First, we
find a global minimum point $q^*$ on $L$.  Second, the point $q^*$
partitions $P$ into two subsets $P_{l}$ and $P_{r}$, for which we compute $S_k(P_{l})$ and $S_k(P_{r})$.
Finally, $S_k(P)$ is obtained by taking the first $k$ points after
merging $S_k(P_{l})$ and $S_k(P_{r})$.

Note that the points in $Q$ may not be given sorted on $L$.
Recall that $W=\sum_{q\in Q}w(q)$.
Let $q^*$ be the point in $Q$ such that
$$\sum_{x(q)<x(q^*), q\in Q}w(q)<W/2 \text{\ \ and\ \  }
w(q^*)+\sum_{x(q)<x(q^*), q\in Q}w(q)\geq W/2.$$ If we
view $w(q)$ as the weight of $x(q)$, then $q^*$ is the {\em weighted median}
of the set $\{x(q)\ |\ q\in Q\}$ \cite{ref:CLRS09}.
We claim that $q^*$ is a global
minimum point on $L$. To prove the claim, we first present Lemma
\ref{lem:1dmonotone}.
We say a function $f(x)$ is {\em monotonically increasing}
(resp., {\em decreasing}) if
$f(x_1)\leq f(x_2)$ for any $x_1\leq x_2$ (resp., $x_1\geq x_2$).
%In this paper ``monotonically increasing'' means ``monotonically
%non-decreasing'', and ``monotonically decreasing'' means ``monotonically
%non-increasing''.

\begin{lemma}\label{lem:1dmonotone}
For any point $p$ on $L$ and $p\neq q^*$, if we move $p$ on $L$
towards $q^*$, the aggregate distance
$\Ed(p)$ is monotonically decreasing.
\end{lemma}
\begin{proof}
Without loss of generality, assume $p$ is on the left side of $q^*$
and we move $p$ on $L$ to the right towards $q^*$. The case where $p$
is on the right side of $q^*$ can be analyzed similarly.
At any moment during the movement of $p$, let $Q_L=\{q \in Q \ | x(q)\leq x(p)\}$ and $Q_R=Q\setminus Q_L$.
According to the definition of $\Ed(p)$, we have
\begin{equation}\label{Eq:1Dspace}
\begin{split}
\Ed(p) & = \sum_{q\in Q}|x(p)-x(q)|=\sum_{q\in Q_L} w(q)\cdot [x(p)-x(q)]+\sum_{q\in Q_R}
w(q)\cdot [x(q)-x(p)]\\
	& = \Big[\sum_{q\in Q_L}w(q)- \sum_{q\in Q_R}w(q)\Big]\cdot x(p)-\sum_{q\in Q_L}w(q)\cdot x(q)+\sum_{q\in
		Q_R}w(q)\cdot x(q).\\
\end{split}
\end{equation}

Because $p$ is to the left of $q^*$, according to the definition
of $q^*$, $\sum_{q\in Q_L}w(q)\leq W/2\leq  \sum_{q\in Q_R}w(q)$ holds.
Further, as $p$ moves to the right towards $q^*$, the value $x(p)$ is
monotonically increasing.

Suppose $p$ is between two points $q_i$ and $q_j$ of $Q$ such that $x(q_i)\leq x(p)<x(q_j)$ and there are no other points of $Q$ between $q_i$ and $q_j$. Note that it is possible that such a point $q_i$ does not exist (i.e., no point of $Q$ is on the left side of $p$), in which case we let $x(q_i)=-\infty$.

If $p$ moves in the interval $[x(q_i),x(q_j))$ to the right, then both sets $Q_L$ and $Q_R$ stay the same, and thus, the value $[\sum_{q\in Q_L}w(q)- \sum_{q\in Q_R}w(q)]\cdot x(p)$ is monotonically decreasing and neither $\sum_{q\in
Q_L}w(q)\cdot x(q)$ nor $\sum_{q\in Q_R}w(q)\cdot x(q)$ changes. Therefore, if $p$ moves in the interval $[x(q_i),x(q_j))$ to the right, $\Ed(p)$ is monotonically decreasing. We claim that for any $p$ in $[x(q_i),x(q_j))$, it always holds that $\Ed(p)\geq \Ed(q_j)$, which leads to the lemma. Indeed, it can be verified that $\Ed(q_j)-\Ed(p)=(x(q_j)-x(p))\cdot [\sum_{q\in Q_L}w(q)- \sum_{q\in Q_R}w(q)]$. Since $x(q_j)-x(p)>0$ and $[\sum_{q\in Q_L}w(q)- \sum_{q\in Q_R}w(q)]\leq 0$, we obtain that $\Ed(q_j)-\Ed(p)\leq 0$.

The lemma thus follows.
\end{proof}

Lemma \ref{lem:1dmonotone} implies that $\Ed(p)$
%is convex function with respect to the position of $p$ on $L$,
attains a global minimum at $p=q^*$. Hence, the point $q^*$ is a
global minimum point on $L$.

Next, we find the set $S_k(P)$ with the help of $q^*$
and Lemma \ref{lem:1dmonotone}.
%Let $P_l$ be the set of points of $P$ to the left of or on $q^*$,
% and let $P_r=P\setminus P_l$.
%We find the points of $S_k(P)$ by scanning the points of $P_l$
%incrementally from right to left and scanning the points of
%$P_r$ incrementally from left to right, as follows.
%Denote the points of $P_l$ ordered from right to left as
%$p^l_1, p^l_2,\ldots, p^l_{|P_l|}$ and those of $P_r$ ordered from left
%to right as $p^r_1, p^r_2,\ldots, p^r_{|P_r|}$. By Lemma \ref{lem:1dmonotone}, the
%$\ANN$ $z_1$ is either $p^l_1$ or $p^r_1$. Without loss of generality, assume
%$\Ed(p^l_1)\leq \Ed(p^r_1)$. Then, $z_1=p^l_1$ and $z_2$ (i.e.,
%the second $\ANN$ of $Q$ in $P$) is either
%$p^r_1$ or $p^l_2$. In general, $S_k(P)$ consists of a subset of
%consecutive points of $P_l$ from $p^l_1$ to the left and a subset of
%consecutive points of $P_r$ from $p^r_1$ to the right.
Let $P_l = \{ p \in P \mid x(p) \leq x(q^{*})\}$ and $P_r=P\setminus
P_l$. We find the set $S_{k}(P_{r})$ of top-$k$ $\ANN$s of $Q$ in $P_{r}$
by scanning the sorted list of $P_r$ from left to right and reporting
the first $k$ scanned points. $S_{k}(P_{l})$ can be obtained
similarly. Among the $2k$ points obtained above, we report the set of
$k$ points with the smallest aggregate distances to $Q$ as $S_k(P)$. We deduce the following theorem.

\begin{theorem}\label{theo:1d}
Given a set $P$ of $n$ points on the real line $L$, with
$O(n\log n)$ preprocessing time and $O(n)$ space, the top-$k$ $\ANN$s
can be found in $O(\min\{k,\log m\}\cdot m+k+\log n)$ time
for any weighted set $Q$ and integer $k$; if the points of $Q$
are given sorted on $L$, then the query time is $O(k+ m+ \log
n)$.
\end{theorem}

\begin{proof}
The only preprocessing is to sort the points in $P$ from left
to right, which takes $O(n\log n)$ time and $O(n)$ space.

Given any query $Q$ and any $k$, we first compute the point $q^*$,
in $O(m)$ time by the weighted selection algorithm \cite{ref:CLRS09}.
The sorted lists of $P_l$ and $P_r$ can be obtained implicitly in $O(\log n)$
time by determining the two neighboring points of $q^{*}$ in the sorted list $P$.
If $k=1$, it is sufficient to consider the two neighboring points of $q^*$ in $P$: we compute, in $O(m)$ time, their  aggregate distances to $Q$, and return the point with smaller  aggregate
distance. Hence, the total time for finding $\ANN$ is $O(m+\log n)$.
Below, we compute $S_k(P)$ for general $k$.

\def\TwoKValues{\Psi}
For simplicity of discussion, we assume $|P_l|\geq k$ and $|P_r|\geq k$.
We first compute the set $S_{k}(P_{r})$ of top-$k$ $\ANN$s of $Q$ in $P_r$ by scanning the sorted list of $P_r$ from left to right and report the first $k$ points. $S_{k}(P_{l})$ can be obtained similarly. Among the found $2k$ points in $S_{k}(P_{l}) \cup S_{k}(P_{r})$, we report the set of $k$ points with the smallest  aggregate distances to $Q$ as $S_k(P)$.
Let $\TwoKValues_l = \{ \Ed(p) \mid p\in S_{k}(P_{l})\}$ denote the set of $k$  aggregate distance values $\Ed(p)$ of all points $p$ in $S_{k}(P_{l})$. Similarly, we define $\TwoKValues_r$. Set $\TwoKValues = \TwoKValues_l \cup \TwoKValues_r$. If we know $\TwoKValues$, the final step can be done easily in $O(k)$ time. $\TwoKValues$ can be computed in $O(mk)$ time in a straightforward way, leading to $O(mk+\log n)$ overall query time. In the following, we show that $\TwoKValues$ can be computed in $O(m\log m+k)$ time, which leads to $O(m\log m+k+\log n)$ overall query time, and further, if $Q$ is given sorted, $\TwoKValues$ can be computed in $O(m+k)$ time.

%We maintain a set $S=\{p_l,p_r\}$ of two points of
%$P$ with $p_l\in P_l$ and $p_r\in P_r$. Initially,
%$S=\{p^l_1,p^r_1\}$. In general, suppose we have found the $k'$-th
%$\ANN$ of $Q$ with $k'\leq k$. If $k'=k$, then we
%are done. Otherwise, the $(k'+1)$-th $\ANN$ of $Q$
%is one of the two points in $S$ with the smaller aggregate distance to
%$Q$. Without loss of generality, assume $\Ed(p_l)\leq \Ed(p_r)$. Then,
%we remove $p_l$ from $S$ and put $p_l$ in $S_k(P)$. Further, we add
%$p_l'$ to $S$, where $p_l'$ is the point of $P_l$ closest to $p_l$ and to the left of
%$p_l$. It remains to discuss how to compute
%$\Ed(p_l)$ and $\Ed(p_r)$. If we compute them in a straightforward way
%in $O(m)$ time, then the entire query algorithm runs in $O(mk+\log n)$ time.
%Below, we give another $O(k+m\log m+\log n)$ time algorithm.

The $m$ points in $Q$ partition $L$ into $m+1$ intervals and an easy
observation is that the  aggregate distance $\Ed(p)$ changes
linearly as $p$ changes in each interval. Specifically, consider
computing $\Ed(p)$ for any given point $p$: if we know the four values $\sum_{q\in Q_L}w(q)$, $\sum_{q\in
Q_R}w(q)$, $\sum_{q\in Q_L}w(q) x(q)$, and $\sum_{q\in
Q_R}w(q) x(q)$ in Eq.~(\ref{Eq:1Dspace}) in the proof of Lemma \ref{lem:1dmonotone}, then $\Ed(p)$ can be computed in constant time. In order to utilize this, we preprocess $Q$ as follows.

We sort the points of $Q$ from left to right and assume the sorted
list is $q_1,q_2,\ldots,q_m$. For each $1\leq j\leq m$, we
compute the four values $\sum_{i=1}^jw(q_i)$,
$\sum_{i=1}^jw(q_i)x(q_i)$, $\sum_{i=j}^m w(q_i)$, and
$\sum_{i=j}^mw(q_i)x(q_i)$. Note that all these $4m$ values can be computed
in $O(m)$ time (after $Q$ is sorted). Then, given any point $p$, if we
know the index $i$ such that $x(q_i)\leq x(p)<x(q_{i+1})$, then
$\Ed(p)$ can be computed in constant time.

Now, we compute $\TwoKValues$. Let $Q_l = \{ q \in Q \mid x(q) \leq x(q^*)\}$ and $Q_r=Q\setminus Q_l$.
After having $q^*$, $Q_l$ and $Q_r$ can
be obtained implicitly in $O(\log m)$ time by binary search on the
sorted list of $Q$. Recall that we scan the points
in $P_r$ from left to right to find $S_k(P_r)$.
If we scan both $P_r$ and $Q_r$ simultaneously, then at the moment of scanning any point, say $p$,
we already know the index $i$ such that $x(q_i)\leq x(p)<x(q_{i+1})$ and thus $\Ed(p)$ can be computed in constant time. Hence, $\TwoKValues_r$ can be computed in $O(m+k)$ time, so can $\TwoKValues_l$. In this way, the total time for computing $\TwoKValues$ is $O(k+m\log m)$. Note that if $Q$ is given sorted on $L$, the query time becomes $O(k+m)$.

The theorem thus follows.
\end{proof}

%Note that the result in Theorem \ref{theo:1d} is applicable for both
%$L_1$ and $L_2$ distances because the two distances are the same in the
%one-dimentional space.

\section{Top-$k$ $\ANN$ Searching in the Plane}
\label{sec:2d}

In this section, we present our results in two-dimensional space,
where the input point set $P$ and the query point $Q$ are given in the
plane.

We generalize the techniques in Section \ref{sec:1d}. For any query $Q$, we first find a
global minimum $q^*$ in the plane.
Then, for each quadrant $R$ of the four quadrants with respect to $q^*$ (i.e., the
four quadrants partitioned by the vertical line and the horizontal
line through $q^*$), we find the top-$k$ $\ANN$s of $Q$ in $P\cap R$ (i.e., $S_k(P \cap R)$)
and compute the aggregate distance values $\Ed(p)$ for all $p \in S_k(P \cap R)$; among the found $4k$
points, we report the set of $k$ points with smallest aggregate distances to $Q$ as $S_k(P)$.
Note that we view each quadrant as a closed region including its
two bounding half-lines (with the common endpoint~$q^*$).

We describe our algorithm for the first quadrant, and the other three
quadrants can be treated in a similar manner. Let $P^1 \subseteq P$ be
the set of points lying in the first quadrant, i.e., $P^1=\{p \in P\
|\ x(p)\geq x(q^*), y(p)\geq y(q^*)\}$.  Our goal is to find
$S_k(P^1)$, the set of top-$k$ $\ANN$s of $Q$ in $P^1$.
Let $z_i$ denote the $i$-th $\ANN$ of $Q$ in $P^1$. Our algorithm computes $S_k(P^1)$ in the order of
$z_1,z_2,\ldots,z_k$.

The problem here is more difficult than that in the 1-D case. For
example, in the 1-D case, for all $i\leq k$, the $i$-th $\ANN$ in $P_l$
or $P_r$ can be easily found by scanning a sorted list. Here, in
contrast, by proving a monotonicity property as Lemma
\ref{lem:1dmonotone}, we show that the $i$-th $\ANN$ $z_i$ in $P^1$ must
be on a ``skyline'' (to be defined later) and we need to somehow search the ``skyline''. After $z_i$ is found, we determine a new
skyline without considering $z_1, \ldots, z_i$, and then find $z_{i+1}$ by searching the new skyline.
This procedure continues until $z_k$ is obtained.
Advanced data structures (e.g., compact interval trees \cite{ref:GuibasCo91} and segment-dragging query data structure
\cite{ref:ChazelleAn88}) are also used for efficient implementations.
%The details are given below.

Consider any query set $Q=\{q_1,q_2,\ldots,q_m\}$ and any $k$.
For any point $p$ in the plane, denote by $x(p)$
the $x$-coordinate of $p$ and by $y(p)$ the $y$-coordinate of $p$.
The aggregate distance of $p$ to $Q$ is $\Ed(p)=\sum_{q\in Q}w(q)\dist(p,q)$,
where $\dist(p,q)=|x(p)-x(q)|+|y(p)-y(q)|$.
Our goal is to find the top-$k$ $\ANN$ set $S_k(P^1)$ in the first quadrant.

\subsection{The Global Minimum Point and the Monotonicity Property}

A point $p$ in the plane is called a {\em global minimum point} if it minimizes the
aggregate distance $\Ed(p)$ among all points in the plane.
Below, we first find a global minimum point and prove a monotonicity property.
Recall that $W=\sum_{q\in Q}w(q)$.  Let $q_x^* \in Q$ be the point such
that
$$\sum_{x(q)<x(q_x^*), q\in Q}w(q)<W/2
\text{\ \ and\ \ } w(q_x^*)+\sum_{x(q)<x(q_x^*), q\in Q}w(q)\geq W/2.$$
If we view $w(q)$ as the weight of $x(q)$ for each $q\in Q$,
then $x(q_x^*)$ is the weighted median of the set $\{x(q) \ |\
q\in Q\}$ \cite{ref:CLRS09}. Similarly,
%We define $q_y^*$ similarly,
 let $q_y^*$ be the point
 in $Q$ such that
 $$\sum_{y(q)<y(q_y^*),\ q\in Q}w(q)<W/2 \text{\ \ and\ \ }
 w(q_y^*)+\sum_{y(q)<y(q_y^*),\ q\in Q}w(q)\geq W/2.$$
We
claim that $q^* = (x(q^*_x), y(q^*_y))$ is a global minimum point.
To prove the claim, we first present Lemma
\ref{lem:2dmonotone}, which generalizes Lemma
\ref{lem:1dmonotone}.
A path in the plane is {\em monotone} if
we move from one endpoint of it to the other,
the $x$-coordinate (resp. $y$-coordinate) is
monotonically changing (either increasing or decreasing).

\begin{lemma}\label{lem:2dmonotone}
For any point $p$ in the plane with $p\neq q^*$,
if we move $p$ towards $q^*$ along a
monotone path, the aggregate distance $\Ed(p)$ is monotonically
decreasing.
\end{lemma}
\begin{proof}
% Without loss of generality, assume $p$ is in the third quadrant with
% respect to $q^*$, i.e., $x(p)\leq x(q^*)$ and $y(p)\leq y(q^*)$.
% Hence, as $p$ moves along any monotone path $\pi$ towards
% $q^*$, both $x(p)$ and $y(p)$ are
% monotonically increasing. The case where $p$ is in other quadrants can
% be analyzed similarly.
According to the definition of $\Ed(p)$, we have
\begin{equation*}
\begin{split}
\Ed(p)
& = \sum_{q\in Q}w(q)\cdot \dist(p,q)
= \sum_{q\in Q}w(q)\cdot \Bigl(|x(p)-x(q)|+|y(p)-y(q)|\Bigr) \\
& = \sum_{q\in Q}w(q)\cdot |x(p)-x(q)| +
\sum_{q\in Q}w(q)\cdot |y(p)-y(q)|\\
& = \Ed_x(x(p))+\Ed_y(y(p)),
\end{split}
\end{equation*}
where
$$\Ed_x(x(p))=\sum_{q\in Q}w(q)\cdot |x(p)-x(q)|\text{ and }\Ed_y(y(p)) = \sum_{q\in Q}w(q)\cdot |y(p)-y(q)|.$$
% Intuitively, $\Ed_x(p)$ is the
% value of $\Ed(p)$ on the $x$-projection and $\Ed_y(p)$ is the
% value of $\Ed(p)$ on the $y$-projection.
If we move $p$ towards $q^*$ along a monotone path, on the $x$-projection, we are moving $x(p)$ towards $x(q^*)$. By Lemma \ref{lem:1dmonotone}, $\Ed_x(x(p))$ is monotonically decreasing, so is $\Ed_y(y(p))$. The lemma thus follows.
% At any moment during the movement, let $Q_L$ be the subset of points
% in $Q$ that are to the left or on the vertical line $x=x(q^*)$,
% i.e., $Q_L=\{q\ |\ q\in Q \text{ and
% } x(q)\leq x(p)\}$. Let $Q_R=Q\setminus Q_L$.
% %Similarly, let $Q_B$ be the subset of points
% %in $Q$ that are below or on the horizontal line $y=y(q^*_y)$, i.e.,
% %$Q_B=\{q\ |\ q\in Q \text{ and
% %} y(q)\leq y(q^*)\}$, and let $Q_A=Q\setminus Q_B$.
% We have
% \begin{equation*}
% \begin{split}
% \Ed_x(p) & = \sum_{q\in Q} w(q)\cdot |x(p)-x(q)|=
% \sum_{q\in Q_L} w(q)\cdot [x(p)-x(q)]+\sum_{q\in Q_R}
% w(q)\cdot [x(q)-x(p)]\\
% & = \Big[\sum_{q\in Q_L}w(q)- \sum_{q\in Q_R}w(q)\Big]\cdot x(p)
% -\sum_{q\in Q_L}w(q)\cdot x(q)+\sum_{q\in Q_R}w(q)\cdot x(q).\\
% \end{split}
% \end{equation*}
%
%
% Recall that $q^*$ is the intersection of the vertical line $x=x(q^*_x)$ and
% the horizontal line $y=y(q^*_y)$.
% Since $x(p)\leq x(q^*_x)$, according to the definition
% of $q_x^*$, $\sum_{q\in Q_L}w(q)\leq W/2\leq  \sum_{q\in Q_R}w(q)$
% always holds.
% Further, as $p$ moves on $\pi$, the value $x(p)$ is
% monotonically increasing.
% Hence, as $p$ moves, the first term above, i.e.,
% $[\sum_{q\in Q_L}w(q)- \sum_{q\in Q_R}w(q)]\cdot x(p)$, is
% monotonically decreasing. As $p$ moves,
% the set $Q_L$ becomes monotonically larger and $Q_R$
% becomes monotonically smaller. Hence, the value
% $\sum_{q\in Q_L}w(q)\cdot x(q)$ is monotonically increasing and
% the value $\sum_{q\in
% Q_R}w(q)\cdot x(q)$ is monotonically decreasing.
% Therefore, as $p$ moves on $\pi$ towards $q^*$,
% the value $\Ed_x(p)$ is monotonically decreasing.
\end{proof}

Lemma \ref{lem:2dmonotone} implies that $\Ed(p)$ attains a global minimum at $p=q^*$.
Hence, the point $q^*$ is a global minimum point in the plane (note that $q^*$ is not necessarily in $Q$).
Next, based on the point $q^*$ and Lemma \ref{lem:2dmonotone}, we
introduce the minimal points and the skyline, and present some observations.

\subsection{The Minimal Points and the Skyline}

We first show how to find $z_1$ (i.e., the $\ANN$ of $Q$ in $P^1$). For any two different points $p_1$ and $p_2$ in $P^1$, we say that $p_1$
{\em dominates} $p_2$ if and only if $x(p_1)\leq x(p_2)$ and $y(p_1)\leq
y(p_2)$. A point $p$ in $P^1$ is called a {\em minimal point} if no point in $P^1$
dominates $p$ (note that the ``minimal'' here is different from the ``global minimum'' defined earlier). If $p_1\in P^1$ dominates $p_2\in P^1$, then there exists a monotone
path $\pi$ connecting $p_2$ and $q^*$ such that $p_1 \in \pi$ (see Fig.~\ref{fig:minpoints}). By
Lemma \ref{lem:2dmonotone}, $\Ed(p_1)\leq \Ed(p_2)$.
Therefore, to compute $z_1$, we only need to consider the set of minimal points in $P^1$, denoted by $M$.
Our discussion above leads to the following lemma.

\begin{figure}[t]
\begin{minipage}[t]{\linewidth}
\begin{center}
\includegraphics[totalheight=1.2in]{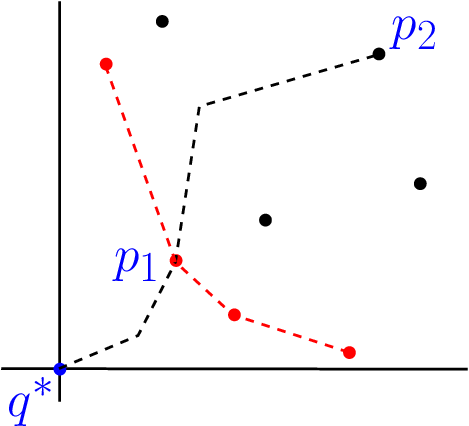}
\caption{\footnotesize The four (red) points connected by the dashed
lines are minimal points, and the dashed line connecting them is a
skyline. $p_1$ dominates $p_2$ and the dotted curve connecting $q^*$
and $p_2$ is a monotone path.}
\label{fig:minpoints}
\end{center}
\end{minipage}
\vspace*{-0.15in}
\end{figure}

\begin{lemma}\label{lem:40}
$z_1 = \underset{p\in M}{\operatorname{argmin}}~\Ed(p)$.
\end{lemma}

One tempting approach is to first find the set $M$ and then find $z_1$.
Unfortunately, here $M$ may have $\Theta(n)$ points
and we cannot afford to check every point of $M$.
Below, we give a better approach.

For each $q\in Q$, we induce a horizontal line and a vertical line
through $q$, respectively; let $\calA$ be the arrangement of the
resulting $2m$ lines. Each cell of $\calA$ is a (possibly unbounded)
rectangle. Each point in $Q$ is a vertex of $\calA$. Note that our
algorithm does not explicitly compute
$\calA$.

We will show below that $\Ed(p)$ is a linear function of $x(p)$ and
$y(p)$ inside any cell $C$ of $\calA$, implying that the $\ANN$ (i.e.,
		the top-1 $\ANN$) of $Q$ in
$P\cap C$ is on the
convex hull of $P\cap C$, as discussed in \cite{ref:AgarwalNe12}.

For any cell $C$, suppose $C = [x_{l}, x_{r}] \times [y_{b}, y_{t}]$.
Set $Q_L = \{ q \in Q \mid x(q) \leq x_{l} \}$, $Q_R = \{ q \in Q \mid x(q) \geq x_{r} \}$, $Q_B = \{ q \in Q \mid y(q) \leq y_{b} \}$ and $Q_T = \{ q \in Q \mid y(q) \geq y_{t} \}$.
According to the construction of $\calA$, no point of $Q$ lies strictly inside~$C$ and $Q=Q_L\cup Q_R =Q_B\cup Q_T$.

We have the following lemma.

\begin{lemma}\label{lem:50}
For any point $p$ in the cell $C$,
$\Ed(p)=C_a\cdot x(p)+C_b\cdot y(p)+C_c$, where
$$C_a=\sum_{q\in Q_L}w(q)-\sum_{q\in Q_R} w(q),\
C_b=\sum_{q\in Q_B}w(q)-\sum_{q\in Q_T} w(q),\ \text{and}$$
$$C_c=\sum_{q\in Q_R}w(q)x(q)-\sum_{q\in Q_L}w(q)x(q)
+\sum_{q\in Q_T}w(q)y(q)-\sum_{q\in Q_B}w(q)y(q).$$
Further, with $O(m\log m)$ time
preprocessing on $Q$, given any cell $C$ of $\calA$, we can compute
$C_a$, $C_b$, and $C_c$ in $O(\log m)$ time.
\end{lemma}
\begin{proof}
The first part (i.e., computing the values of $C_a$, $C_b$, and $C_c$)
has been discussed in \cite{ref:AgarwalNe12}
and it can also be easily verified by our analysis in Lemma \ref{lem:2dmonotone}.
Hence, we omit the proof for it.

For the second part, given any cell $C$, our goal is
to compute the three values $C_a$, $C_b$, and $C_c$.
Generally speaking, if, as preprocessing,
we compute the prefix sums of the values $w(q)$
and $w(q)x(q)$ in the sorted list of the points of $Q$ by
their $x$-coordinates, and compute the prefix sum of $w(q)y(q)$
in the sorted list of the points of $Q$ by
their $y$-coordinates, then $C_a$, $C_b$, and $C_c$ can be computed
in $O(\log m)$ time. The details are given below.

To compute $C_a$, we need to know the value $\sum_{q\in Q_L}w(q)$ and
the value $\sum_{q\in Q_R} w(q)$. Note that $\sum_{q\in Q_R}
w(q)=W-\sum_{q\in Q_L}w(q)$.
We can do the following preprocessing.
We sort all points in $Q$ by their $x$-coordinates. Suppose the
sorted list is $q_1,q_2,\ldots,q_m$ from left to right. For each $1\leq j\leq m$, we compute the value $W_1(q_j)=\sum_{i=1}^jw(q_i)$.
For any given cell $C$, let $x_l$ be the $x$-coordinate of the
vertical line containing the left side of $C$. By binary search on the
sorted list $q_1,q_2,\ldots,q_m$, in $O(\log m)$ time,
we can find the rightmost
point $q'$ in $Q$ such that $x(q')\leq x_l$. It is easy to see that
$\sum_{q\in Q_L} w(q)=W_1(q')$.
Note that the above preprocessing takes $O(m\log m)$ time, and $C_a$
can be computed in $O(\log m)$ time.

In similar ways, we can compute $C_b$ and $C_c$ in $O(\log m)$ time,
with $O(m\log m)$ preprocessing  time.
Hence, the second part of the lemma follows.
\end{proof}

As discussed in \cite{ref:AgarwalNe12}, Lemma \ref{lem:50} implies
that $S_1(P\cap C)$ (i.e., the $\ANN$ of $Q$ in $P\cap C$) is on the
convex hull of $P\cap C$. More specifically, $S_1(P\cap C)$ is an
extreme point of $P\cap C$ along a certain direction that is
determined by $C_a$ and $C_b$, and thus we can do binary search on
the convex hull to find it.

To compute $z_1$, the algorithm in \cite{ref:AgarwalNe12} checks
every cell $C$ of $\calA$ in the first quadrant, and it finds $S_1(P\cap C)$ by doing binary search on the
convex hull of the points in $P\cap C$.  The number of cells checked in
\cite{ref:AgarwalNe12} is $O(m^2)$.  In contrast, based on
Lemma \ref{lem:40}, we show below that we only need to check $O(m)$ cells.
%In addition, we use the compact interval trees \cite{ref:GuibasCo91}
%to (implicitly) compute the convex hulls in a faster way than that in
%\cite{ref:AgarwalNe12}.
Although the number of minimal points in $M$ can be $\Theta(n)$, we
show that the number of cells of $\calA$ that contain these
minimal points is $O(m)$, and further, we can find these cells efficiently.

If we order the points in $M$ by their $x$-coordinates and connect
every pair of adjacent points by a line segment, then we can obtain a
path $\pi_1$, which we call a  {\em skyline} (see Fig.~\ref{fig:minpoints}).
%Let $p_l$ be the leftmost point in $M$ and $p_r$ be the
%rightmost point in $M$.
The points of $M$ are also considered as the {\em vertices} of $\pi_1$.
If we move on $\pi_1$ from its left endpoint to its right endpoint,
then the $x$-coordinate is monotonically increasing and the
$y$-coordinate is monotonically decreasing. Hence, $\pi_1$ is a monotone path.

Denote by $\calC_1$ the set of cells of $\calA$ that contain the
minimal points in $M$.

\begin{lemma}\label{lem:60}
$|\calC_1| = O(m)$.
%The number of cells of $\calA$ containing the minimal points in $M$ is~$O(m)$.
\end{lemma}
\begin{proof}
Due to our general position assumption that no two points in $P\cup Q$ have
the same $x$-coordinate or $y$-coordinate. Each edge of $\pi_1$ is
neither horizontal nor vertical.
Because $\pi_1$ is a monotone path, each line of $\calA$ can intersect
$\pi_1$ at most once. Hence, the number of intersections between $\pi_1$
and $\calA$ is $O(m)$, which implies that
the number of cells that intersect $\pi_1$ is
$O(m)$. Since all points in $M$ are on $\pi_1$, the lemma follows.
\end{proof}

 Due to Lemma \ref{lem:40}, the following lemma
is obvious.

\begin{lemma}\label{lem:65}
The point $z_1$ is in one of the cells of $\calC_1$.
\end{lemma}

\subsection{Computing the Set $\calC_1$}

Next, we show how to compute $\calC_1$.
A straightforward way is to first
compute $\calA$ and then traverse $\calA$ by following the skyline
$\pi_1$. But this approach is not efficient due to: (1)
computing $\calA$ takes $\Theta(m^2)$ time; (2)  the size of $\pi_1$ may be
$\Theta(n)$ due to $|M|=\Theta(n)$ in the worst case.
%Hence, the approach may runs in $\Theta(k^2+n)$ time in the worst case.
Below in Lemma \ref{lem:70},
we propose to compute~$\calC_1$ in $O(m\log n+m\log m)$ time.

First of all, we sort all points in $Q$ by their $x$-coordinates and
$y$-coordinates, respectively; accordingly, we obtain a sorted list for the
horizontal lines of $\calA$ and a sorted list for the vertical lines
of $\calA$. With these two sorted lists, given any point $p$, we can
determine the cell of $\calA$ that contains $p$ in $O(\log m)$ time by
doing binary search on the two sorted lists. We should point out that
there might be other ways to compute $\calC_1$, but the algorithm we propose
for Lemma \ref{lem:70} is particularly useful later when we compute other points in $S_k(P^1)$ than $z_1$.

\begin{lemma}\label{lem:70}
$P$ can be preprocessed in $O(n\log n)$ time using $O(n)$ space, such that given any $Q$,  we can compute
the set $\calC_1$ in $O(m\log n+m\log m)$ time.
\end{lemma}
\begin{proof}
One operation frequently used for computing $\calC_1$
is the following {\em segment-dragging queries}.
Given any horizontal or vertical line segment $s$, we move
$s$ along a given direction perpendicular to $s$; the query
asks for the first point of $P$ hit by $s$ or reports no such point exists.
Chazelle \cite{ref:ChazelleAn88} constructed an $O(n)$-size data
structure in $O(n\log n)$ time such that each segment-dragging query
can be answered in $O(\log n)$ time.
As preprocessing, we build such a data structure on $P$.
Hence, the preprocessing takes $O(n\log n)$ time and $O(n)$ space.

For each cell $C$ of $\calC_1$, we call the leftmost point of $M\cap
C$ the {\em skyline-left point} of $C$ and call the bottommost point
of $M\cap C$  the {\em skyline-bottom point} of $C$.
In other words, if we move along the skyline $\pi_1$ from its left
endpoint to its right endpoint, then the skyline-left point of $C$ is
the first vertex of $\pi_1$ we meet in $C$ and the skyline-bottom point of $C$
is the last vertex of $\pi_1$ we meet in $C$. Note that if $C$ has only one
minimal point of $M$, then the only minimum point is both the skyline-left
point and the skyline-bottom point of $C$.

We will find the skyline-left point and the skyline-bottom point for each cell $C\in
\calC_1$. Each such point $p$ is determined by a segment-dragging
query on a segment $s$ and we call $s$ the {\em generating segment} of
$p$; $s$ will be associated with $p$ for later use (for computing
other points in $S_k(P^1)$ than $z_1$). Further, we will classify
these generating segments into four types, and again, they will be
useful later in Lemma \ref{lem:100} for computing $S_k(P^1)$.

All the vertical lines passing through points in $Q$ partition the space into $O(m)$ regions, which we refer to as {\em
columns} (including bounding lines).
Let $\calD_M$ denote the set of columns of $\calA$ each of which contains at
least one cell of $\calC_1$.
We search the columns of $\calD_M$ from left to right.
For each column $D\in \calD_M$, we will first find the topmost cell and the bottommost cell of $\calC_1$ in
$D$; then, from the bottommost cell to the topmost cell, we search all other
cells of $\calC_1$ in $D$ in a bottom-up fashion.
%More specifically, we first find the lowest cell of
%$\calC_1\cap D$ and then find the second lowest cell of $\calC_1\cap
%D$. This searching procedure is done until the highest cell of
%$\calC_1\cap D$ cell is found.
After the searching on $D$ is done, we proceed to the next column of $\calD_M$.
The details are given below.

Note that due to the general position assumption that no two points
in $P\cup Q$ have the same $x$- or $y$-coordinate, each point of $P$
lies strictly inside a cell of $\calA$.

We first determine the leftmost column of $\calD_M$, denoted by $D$, which
is the one containing the leftmost point $p_0$ of $M$ (see Fig.~\ref{fig:searchcells}).
$p_0$ can be found by the following segment-dragging query.
Let $y_{max} = \max_{p \in P^1} y(p)$.
Consider a vertical segment $s_0=\overline{q^*b}$ where
$b = (x(q^*), y_{max})$. If we drag $s_0$ rightwards (i.e., horizontally
to the right), $p_0$ will be the first point of $P^1$ hit by
$s_0$. By using the segment-dragging query data structure
on $P$, $p_0$ can be found in $O(\log n)$ time. After having $p_0$, $D$
can be determined in $O(\log m)$ time using binary search on the sorted list of the vertical lines of $\calA$.

\begin{figure}[t]
\begin{minipage}[t]{\linewidth}
\begin{center}
\includegraphics[totalheight=2.0in]{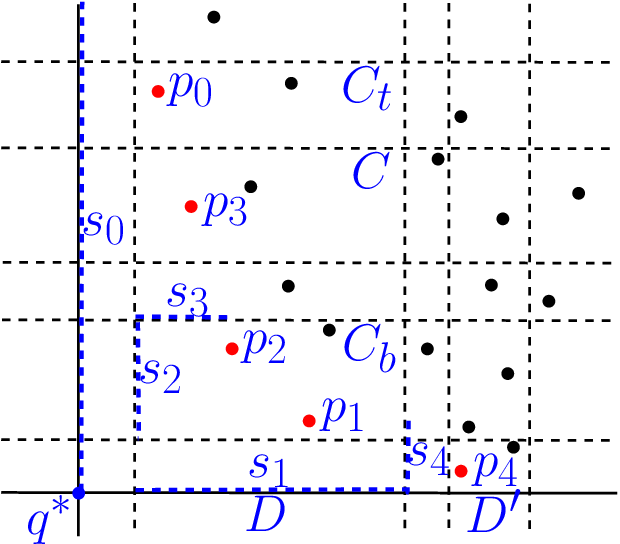}
\caption{\footnotesize Illustrating the algorithm in Lemma
\ref{lem:70}: the dashed grid is $\calA$.
}
\label{fig:searchcells}
\end{center}
\end{minipage}
\vspace*{-0.15in}
\end{figure}

Notice that the cell of $\calA$ that contains $p_0$ is the
topmost cell in $D\cap\calC_1$, which we denote by $C_t$, and that $p_0$ is the skyline-left point
of $C_t$ (see Fig.~\ref{fig:searchcells}). The segment $s_0$ is
the generating segment of $p_0$ and we classify $s_0$ as an {\em
$s_0$-type} generating segment. In general, the $s_0$-type generating segments are used to
find the skyline-left points of the topmost cells of the columns of
$\calD_M$.

Next, we determine the bottommost cell of $D\cap\calC_1$, denoted by $C_b$. We first determine
the skyline-bottom point $p_1$ of $C_b$ by a segment-dragging query as
follows. Let $\ell$ denote the horizontal line $y = y(q^*)$. Set $s_1 = D \cap \ell$.
If we drag $s_1$ upwards, $p_1$ will be the first point of $P^1$ hit
by $s_1$ (see Fig.~\ref{fig:searchcells}). After $p_1$ is found in $O(\log n)$ time, $C_b$ can be determined in additional $O(\log m)$ time. $s_1$ is the generating segment of $p_1$ and we classify
$s_1$ as the {\em $s_1$-type} generating segment. In general, $s_1$-type
generating segments are used to
find the skyline-bottom points of the bottommost cells of the columns of
$\calD_M$.

If $C_b = C_t$, then the column $D$ contains only one cell of $\calC_1$, and our searching on $D$ is done.
Below, we assume $C_b \neq C_t$.

\def\SecCell{C_s}
In the sequel, from the bottommost cell $C_b$,
we search the cells of $\calC_1$ in $D$ in a bottom-up manner until we
meet the topmost cell $C_t$.
We first show how to determine the second lowest cell of $\calC_1\cap
D$ (i.e., the one of $\calC_1\cap D$ right above $C_b$), denoted by $\SecCell$.

To determine $\SecCell$, we first find the skyline-left point $p_2$ of $C_b$ using a segment-dragging query, as
follows. Let $s_2$ be the left side of $C_b$. $p_2$ is the first point
in $P^1$ hit by dragging $s_2$ rightwards (see
Fig.~\ref{fig:searchcells}). $s_2$ is the generating segment of
$p_2$ and we classify $s_2$ as the {\em $s_2$-type} generating
segment. In general, each $s_2$-type generating segment is used to find
the skyline-left point of a cell whose skyline-bottom point has just
been found. Next, we determine $\SecCell$ by using $p_2$.

We first determine the skyline-bottom point $p_3$ of $\SecCell$.
Since $y(p_3) > y(p_2)$, $x(p_3)<x(p_2)$ (otherwise $p_2$ would dominate $p_3$).
An easy observation is that $p_3$ is the lowest point among all points
of $P^1\cap D$ whose $x$-coordinates are less than $x(p_2)$ (see
Fig.~\ref{fig:searchcells}).
We can determine $p_3$ by the following segment-dragging query.
%The vertical line through $p_2$ partitions the column $D$
%into two vertical ``sub-columns'', and denote by $D_l$ the left
%sub-column. Let $p_3$ be the lowest point in $P\cap D_l$ (e.g., see
%Fig.~\ref{fig:searchcells}). Let $C''$ be
%the cell containing $p_3$. We claim that $C''$ is is $C'$. We prove
%the claim in the next paragraph.
%
%Indeed, since $x(p_3)<x(p_2)$ and $p_2$ is the leftmost point in $P\cap C$,
%$C''$ cannot be $C$.
%Thus, $C''$ is higher than $C$.
%On the other hand, suppose to the contrary that $C''$ is not $C'$.
%Then, $C'$ is above $C$ and below $C''$.  Also,
%$C'\cap D_l$ must contain a point of $M$ since
%otherwise all minimal points in $C'\cap M$ are dominated by $p_2$,
%contradicting with that $C'\in \calC_1$ contains minimal points of $M$.
%Since $C'$ is lower than $C''$ and $C'\cap D_l$ contains minimal points of
%$M$, this contradicts with that $p_3\in C''$ is the lowest
%point in $P\cap D_l$. Hence, we conclude that $C''$ is $C'$.
Let $s_3$ be the horizontal line segment on the top side of the cell
$C_b$ such that the left endpoint of $s_3$ is the upper left vertex of
$C_b$ and the right endpoint has $x$-coordinate $x(p_2)$ (see
Fig.~\ref{fig:searchcells}).
%Suppose $C_{b} = [x'_{l}, x'_{r}] \times [y'_{b}, y'_{t}]$.
%Set $s_{3} = [x'_{l}, x(p_{2})]\times\{y'_{t}\}$.
Due to our general position assumption that no two points in $P\cup Q$
have the same $x$- or $y$-coordinate, $p_3$ is the point of $P^{1}$ hit first by
dragging $s_3$ upwards. After $p_3$ is found, $\SecCell$ can be
determined.
Therefore, we can determine $\SecCell$ in $O(\log n+\log m)$ time.
$s_3$ is the generating segment of $p_3$ and we
classify $s_3$ as the {\em $s_3$-type} generating segments. In
general, $s_3$-type generating segments are used to find the
skyline-bottom points for non-bottommost cells of the columns of
$\calD_M$.

If $\SecCell=C_t$, we are done searching on $D$.
Otherwise, we continue the above procedure to search other
cells of $\calC_1\cap D$ until we meet the topmost cell $C_t$.

Now we proceed to the next column $D' \in \calD_M$, in the following way.
We first determine $D'$ by a segment-dragging query as follows.
Recall that $p_1$ is the lowest point in $P^1\cap D$.
%Let $x = x_r$ denote the right bounding line of $D$.
%Set $s_4 = \{x_r\} \times [y(q^*), y(p_1)]$.
Let $s_4$ be the vertical line segment on the right bounding line of
$D$ such that the lower endpoint of $s_4$ has $y$-coordinate $y(q^*)$
and the upper endpoint has $y$-coordinate $y(p_1)$ (see Fig.~\ref{fig:searchcells}).
We drag the segment $s_4$ rightwards, and let $p_4$ be the first point of
$P^1$ hit by $s_4$ (see Fig.~\ref{fig:searchcells}).
It is not difficult to see that $p_4$ is
a minimal point and the
column of $\calA$ containing $p_4$ is $D'$. Further, $p_4$ is the
skyline-left point of the topmost
cell of $\calC_1\cap D'$.
Hence, after $p_4$ is found, $D'$ and the topmost cell of $\calC_1\cap
D'$ can be determined in $O(\log m)$ time. $s_4$ is the
generating segment of $p_4$; note that $s_4$ is an $s_0$-type
generating segment.

Note that if the above segment-dragging query on $s_4$ fails to
find any point (i.e., such a point $p_4$ does not exist), then all
cells of $\calC_1$ have been found, and we terminate.
Otherwise, we proceed to search all cells in $\calC_1\cap D'$ in the same
way as in the column $D$, and then search other columns of $\calD_M$
similarly.

For the running time, as shown above, the
algorithm spends $O(\log n+\log m)$ time finding each cell of
$\calC_1$.
Due to $|\calC_1|=O(m)$ (by Lemma \ref{lem:60}),
computing $\calC_1$ takes $O(m\log n+m\log m)$ time. The lemma thus follows.
\end{proof}

\subsection{Computing the Top-$k$ $\ANN$ Set $S_k(P^1)$}

In this section, we compute $S_k(P^1)$ in the order of
$z_1,z_2,\ldots,z_k$.

Since $z_1$ is in one of the cells of $\calC_1$, once we have
$\calC_1$, we compute the $\ANN$ of $Q$ in $C\cap P$ in each cell $C\in
\calC_1$; among the $|\calC_1|$ candidate points,
$z_1$ is the one with the smallest aggregate distance to
$Q$. Once $z_1$ is obtained, we use a similar approach to compute $z_2$.
Let $\pi_2$ be the skyline of $P^1\setminus\{z_1\}$, and let $\calC_2$
be the set of cells of $\calA$ that contain the vertices of $\pi_2$.
Again, $z_2$ must be in one of the cells of
$\calC_2$, and we find $z_2$ by searching the cells of $\calC_2$. In
general, let $\pi_{i}$ be the skyline of
$P^1\setminus\{z_1,\ldots,z_{i-1}\}$, and let $\calC_{i}$
be the set of cells of $\calA$ that contain the vertices
of $\pi_{i}$. The point $z_{i}$ must be in one of the cells of
$\calC_{i}$, and we find $z_{i}$ by searching the cells of
$\calC_{i}$. We repeat this till $z_k$ is found.

For each $1\leq i\leq k$, since $\pi_i$ is a skyline,
$|\calC_i|=O(m)$. A straightforward
implementation to compute $S_k(P^1)$ will need to search $O(km)$ cells. We will show that we
only need to search $O(k+m)$ cells in total, and more importantly, we
can find all these cells efficiently. Specifically, we propose an
algorithm that can efficiently determine the set $\calC_i$
by updating the set $\calC_{i-1}$, for all $2\leq i\leq k$.

In the sequel, we first present an algorithm that can quickly compute
the $\ANN$ of $Q$ in $C\cap P$ for any cell $C$ of $\calA$.
An $O(n\log^2 n)$-size data structure was given
in \cite{ref:AgarwalNe12} that can be built in $O(n\log^2 n)$ time and
can compute the $\ANN$ in any cell $C$ of $\calA$ in $O(\log^3 n)$ time.
By using compact interval trees \cite{ref:GuibasCo91},
we have the following improved result in Lemma \ref{lem:80}.

\begin{lemma}\label{lem:80}
For a set $P$ of $n$ points in the plane, an $O(n\log n\log\log n)$-size data structure can be built in
$O(n\log n\log\log n)$ time, such that given any axis-parallel rectangle $C$ (e.g., any cell of $\calA$), the $\ANN$ of $Q$ in $P\cap C$ can be computed in $O(\log^2 n)$ time. With trade-off between preprocessing and query time, we can build two other data structures: the first one has $O(n\log n)$ preprocessing time and space with $O(\log^2 n\log\log n)$ query time; the second one has $O(n\log n\log^* n)$ preprocessing time and space with $O(\log^2 n\log^* n)$ query time.
\end{lemma}
\begin{proof}
Our data structure uses the compact interval tree
\cite{ref:GuibasCo91}, which was for solving the following {\em
sub-path hull queries} in \cite{ref:GuibasCo91}.
Let $\pi$ be a simple path of $n$ vertices in the plane and
suppose the vertices are $v_1,v_2,\ldots,v_n$ ordered along $\pi$.
Given two vertex indices $i$ and $j$ with $i < j$,
the {\em sub-path hull query} asks for the convex hull of all vertices
$v_i,v_{i+1},\ldots,v_j$.
A compact interval tree data structure was given
in \cite{ref:GuibasCo91}, and for each
sub-path hull query, it can report in $O(\log n)$ time
a data structure that represents the convex hull
such that any standard binary-search based operation
on the convex hull can be implemented in $O(\log n)$ time (e.g.,
finding an extreme point on the convex hull along any given
direction). Assume the vertices of $\pi$ are sorted by their $x$- or $y$-coordinates.
The compact interval tree is of $O(n\log\log n)$ size and can
be built in the same time. With trade-off between preprocessing and query time, two other compact interval trees can be built for the sub-path hull queries: the first one has $O(n)$ preprocessing time and space with $O(\log n\log\log n)$ query time; the second one has $O(n\log^* n)$ preprocessing time and space with $O(\log n\log^* n)$ query time.

Our data structure for the lemma is constructed as follows.
At the high-level, it is similar to the two-dimensional orthogonal range tree
\cite{ref:deBergCo08}. A balanced binary search tree $T$ is built
based on the
$x$-coordinates of the points in $P$. The leaves of $T$ store the
points of $P$ in sorted order from left to right, and the internal
nodes store splitting values to guide the search on $T$. For each node
$v$ of $T$, it also stores the subset $P(v) \subseteq P$ of points in the subtree
of $T$ rooted at $v$, and $P(v)$ is called the {\em canonical subset}
of $v$. For each canonical subset $P(v)$, we build a compact interval
tree in the following way. If we sort the points of $P(v)$ by their
$y$-coordinates and connect each pair of adjacent points in the
sorted list by a line segment, we obtain a path $\pi(v)$.
The points in $P(v)$ are vertices of $\pi(v)$.
Note that $\pi(v)$ is a simple path and
each horizontal line intersects $\pi(v)$ at most once.
We build a compact interval tree data structure on $\pi(v)$ using
the approaches in \cite{ref:GuibasCo91}.
This finishes the construction of our data structure.

For each canonical subset $P(v)$, depending on which of the three compact interval trees is used, constructing the compact interval tree
on $\pi(v)$ takes $O(\mu\log\log \mu)$ (or $O(\mu)$ or $O(\mu\log^*\mu)$) time and space, where $\mu=|P(v)|$. Note that the $y$-sorted list of
$P(v)$ can be built during the construction of $T$ in a bottom-up manner.
Hence, the preprocessing time and space is $O(n\log n\log\log n)$ (or $O(n\log n)$ or $O(n\log n\log^* n)$), the same as claimed in the lemma statement.

Given any axis-parallel rectangle $C$, our goal is to find the
$\ANN$ of $Q$ in $C\cap P$. Essentially, we are looking for an
extreme point in $C\cap P$ along a certain direction, denoted by $\sigma$. As discussed
in \cite{ref:AgarwalNe12}, $\sigma$ is determined by
the two factors $C_a$ and $C_b$ defined in Lemma \ref{lem:50}, and can be computed in $O(\log m)$ time by Lemma \ref{lem:50}. Recall that we have assumed that $m<n$ holds; hence $\log m=O(\log n)$. 

Suppose $C =  [x_l,x_r] \times [y_{b}, y_{t}]$. Using the range $[x_l,x_r]$, we first find the $O(\log n)$ canonical
subsets whose union is the set of points in $P$ lying between the two
vertical lines $x=x_l$ and $x=x_r$.  For each such canonical subset
$P(v)$, we use the range $[y_b,y_t]$ to determine the sub-path
of $\pi(v)$ inside $C$, which can be done by binary search on the
$y$-sorted list of $P(v)$; subsequently, we use the compact interval
tree data structure on $\pi(v)$
to (implicitly) report the convex hull of the sub-path in $O(\log n)$ time (or $O(\log n\log\log n)$ or $O(\log n\log^* n)$ time), after
which we search the extreme point on the convex hull
along the direction $\sigma$ in $O(\log n)$ time.
In this way, we obtain $O(\log n)$ extreme points for these $O(\log n)$
canonical subsets, and the one minimizing the aggregate distance
to $Q$ is the $\ANN$ of $Q$ in $C\cap P$.
Assuming that we have computed the three factors
$C_a$, $C_b$, and $C_c$ as defined in Lemma \ref{lem:50}, for each extreme
point found above, its aggregate distance to $Q$ can be computed in
constant time.

Therefore, the $\ANN$ of $Q$ in $C\cap P$ can be found
in $O(\log^2 n)$ time (or $O(\log^2 n\log\log n)$ or $O(\log^2 n\log^* n)$ time).
The lemma thus follows.
\end{proof}

In the following, to avoid tedious discussions, unless otherwise stated, when we refer to Lemma \ref{lem:80} we will always use the data structure with $O(\log^2n)$ query time, with the understanding that using different data structures will give different performances (i.e., preprocessing and query time) accordingly.

Let $m_i=|\calC_i|$ for each $1\leq i\leq k$. By Lemma \ref{lem:80},
we can determine $z_1$ in $O(m_1\log^2 n)$ time.  Next we continue to
compute $z_2$. To this end, we need to find the set $\calC_2$ first.
Instead of computing $\calC_2$ from scratch as we did for $\calC_1$,
we obtain $\calC_2$ by updating $\calC_1$. Specifically, if some
cells are both in $\calC_1$ and $\calC_2$, we do not need to
compute them again. In other words, we only need to compute the cells
in $\calC_2\setminus \calC_1$. Let $C(z_1)$ denote the cell containing
$z_1$. In fact, we will show that all the cells of $\calC_1$ except
$C(z_1)$ must be in $\calC_2$. The cell $C(z_1)$ may or may not be in
$\calC_2$. If $C(z_1) \in \calC_2$, then special care needs to be taken
when searching $C(z_1)$ because we are looking for $z_2$ and
the point $z_1$ should not be considered any more.
The details are given below.

For each $2\leq i\leq k$, let $\calC_i'=\calC_i\setminus\calC_{i-1}$
and $m_i'=|\calC_i'|$.
We first show that $\calC_2$ can be obtained in $O(m'_2(\log^2 n+\log m))$ time, and specifically,
we compute the cells of $\calC_2'$ and determine whether $C(z_1)\in\calC_2$, which is done in Lemma \ref{lem:100}.

The algorithm in Lemma \ref{lem:100} needs a dynamic version of the segment-dragging
query data structure that can support point deletions and insertions for $P$.
Later after we finish the query, we also need to insert those points that have been deleted back to $P$, and we call them {\em special insertions}, i.e., whenever we insert a point $p$ to $P$ for the segment-dragging query data structure, $p$ has already been deleted from $P$ before.
In the following Lemma \ref{lem:dynamic}, we present such a data
structure by using the range trees \cite{ref:deBergCo08}.
Note that the performance of the data structure in Lemma
\ref{lem:dynamic} may not be the best: Since other parts of our
algorithm for computing $S_k(P^1)$ dominate the overall running time, we
choose to present a data structure that is simple and does not affect
the overall performance.

\begin{lemma}\label{lem:dynamic}
For a set $P$ of $n$ points in the plane, we can build a data
structure in $O(n\log n)$ time and $O(n\log n)$ space that can answer
each segment-dragging query in $O(\log^2 n)$ time and support each
point deletion and special insertion for $P$ in $O(\log^2 n)$ time.
\end{lemma}
\begin{proof}
Our data structure consists of two range trees, one for horizontal
segment-dragging queries and the other for vertical segment-dragging
queries. Below, we only present the one for horizontal
segment-dragging queries and the other one can be obtained similarly.

We first sort the
points in $P$ by their $x$-coordinates and $y$-coordinates,
respectively. We build a balanced binary search tree $T$ based on the
$x$-coordinates of the points in $P$. The leaves of $T$ store the
points of $P$ in sorted order from left to right. Each node $v$ of
$T$ also stores the subset $P(v)$ of points stored in the leaves
of the subtree rooted at $v$; $P(v)$ is called the {\em canonical
subset} of $v$. For each node $v$, we use another balanced binary
search tree $T(v)$ to store the points in $P(v)$ based on the
$y$-coordinates of the points. It is commonly known that $T$
can be constructed in $O(n\log n)$ time using $O(n\log n)$ space
\cite{ref:deBergCo08}.

Consider any segment-dragging query. Without loss of generality,
assume we drag upwards a horizontal segment $s = [x_1(s),x_2(s)] \times \{y(s)\}$
(i.e., its $y$-coordinate is $y(s)$ and
its $x$-coordinate spans the interval $[x_1(s),x_2(s)]$).
We first determine the $O(\log n)$
canonical subsets of $T$ whose union is the subset of points of $P$
with $x$-coordinates lying in $[x_1(s),x_2(s)]$.
For each canonical subset $P(v)$, we use the tree $T(v)$ to determine
in $O(\log n)$ time
the lowest point of $P(v)$ whose $y$-coordinate is no less than $y(s)$
and that point will be the first point hit by dragging $s$
upwards. After we find such a
point in each canonical subset, we report the point
with smallest $y$-coordinate as
the answer to the segment-dragging query for $s$.
The total query time is $O(\log^2 n)$ time.

%Now consider deleting a point $p$ from $P$. We can simply remove the
%leave of $T$ storing $p$. Further, for each ancestor $v$ of the leave
%storing $p$, we delete $p$ from the tree $T(v)$, which can be done in
%$O(\log n)$ time. Hence, it takes $O(\log^2 n)$  time for each
%point deletion. Point insertions can be done symmetrically.
%Note that it will be seen later that the points that are inserted are exactly those
%points that have been deleted, and therefore, we do not
%need to do ``rotation'' operations when deleting or inserting a point because the
%height of $T$ will always be bounded by $O(\log n)$.

Now consider deleting a point $p$ from $P$. We first find the leaf $v_p$ of $T$ storing $p$. Then, for each node $v$ in the path of $T$ from the root to $v_p$, we delete $p$ from the tree $T(v)$, which can be done in $O(\log n)$ time.
Hence, it takes $O(\log^2 n)$  time for each point deletion.
Consider a special insertion that inserts a point $p$ to $P$. Since it is a special insertion, $p$ was in $P$ before but has been deleted. We first find the leaf $v_p$ of $T$ that stored $p$ before. Then, for each node $v$ in the path of $T$ from the root to $v_p$, we insert $p$ to the tree $T(v)$, which can be done in $O(\log n)$ time.
Hence, it takes $O(\log^2 n)$  time for each special insertion.

The lemma thus follows.
\end{proof}

Next in Lemma \ref{lem:100}, we compute $\calC_2$ based on $\calC_1$, by using the data structure in
Lemma \ref{lem:dynamic}.
The algorithm for Lemma \ref{lem:100} essentially follows the behavior
of the algorithm for Lemma \ref{lem:70}, but only focuses on searching
the cells of $\calC_2'$. The efficiency of the algorithm for Lemma
\ref{lem:100} also hinges on the observation that the cells of
$\calC_2'$ form at most two subsets (separated by $C(z_1)$ if
$C(z_1) \in \calC_2$) of consecutive cells of $\calC_2$ if we
order the cells of $\calC_2$ from ``northwest'' to ``southeast''.

\begin{lemma}\label{lem:100}
We can determine the set $\calC_2$ in $O((1+m'_2)(\log^2 n+\log m))$ time,
where $m'_2=|\calC'_2|$, and more specifically,
our algorithm will compute the cells of $\calC_2'$ and determine
whether $C(z_1)\in \calC_2$.
\end{lemma}
\begin{proof}
We call the order of the cells of $\calC_1$ by which the skyline
$\pi_1$ crosses them from left to right the {\em canonical order} of
$\calC_1$. In other words, the canonical order of $\calC_1$ follows
the northwest-to-southeast order.
We define the canonical order of $\calC_2$ similarly.

Suppose the canonical order
of the cells of $\calC_1$ is: $C_1,C_2,\cdots,C_{m_1}$.
Note that we can obtain
this ordered list during computing $\calC_1$ in Lemma \ref{lem:70} within
the same running time. Recall that when computing $\calC_1$ we also
computed a skyline-left point and a skyline-bottom point for each cell
of $\calC_1$ as well as their generating segments.
Let $C_i=C(z_1)$, i.e., the cell that contains $z_1$.
We assume $i\neq 1$ and $i\neq m_1$ (otherwise the algorithm is
similar and much simpler).
%Let $\calC'_2=\calC_2\setminus \calC_1$.

In order to better understand the algorithm we will present below,
we first discuss a question:
which cells are possibly in $\calC'_2$? Imagine that we partition the
plane into four quadrants with respect to $z_1$ by the vertical line
through $z_1$ and the horizontal line through $z_1$; an easy
observation is that only the cells intersecting the first quadrant
can possibly be in $\calC'_2$ because only points in the first
quadrant are dominated by $z_1$. Further, for each cell $C_j$ with $j\neq
i$, none of the vertices of the skyline $\pi_1$ in $C_j$ is
dominated by $z_1$,
and thus $C_j$ is still in $\calC_2$.
In other words, all cells of
$\calC_1\setminus\{C_i\}$ are still in $\calC_2$. The cell $C_i$ may or
may not be in $\calC_2$. Also note that if we remove $z_1$ from $P$,
then the skyline-bottom point of
$C_{i-1}$ may be changed (see
Fig.~\ref{fig:newskyline}), but the
skyline-left point of $C_{i-1}$ does not change; for each cell $C_j$
with $1\leq j\leq i-2$, neither its skyline-left point nor its
skyline-bottom point changes. Similarly, due to the removal of $z_1$,
the skyline-left point of $C_{i+1}$ may be changed, but its
skyline-bottom point does not change; for each cell $C_j$
with $i+2\leq j\leq m_1$, neither its skyline-left point nor its
skyline-bottom point changes.

\begin{figure}[t]
\begin{minipage}[t]{\linewidth}
\begin{center}
\includegraphics[totalheight=2.0in]{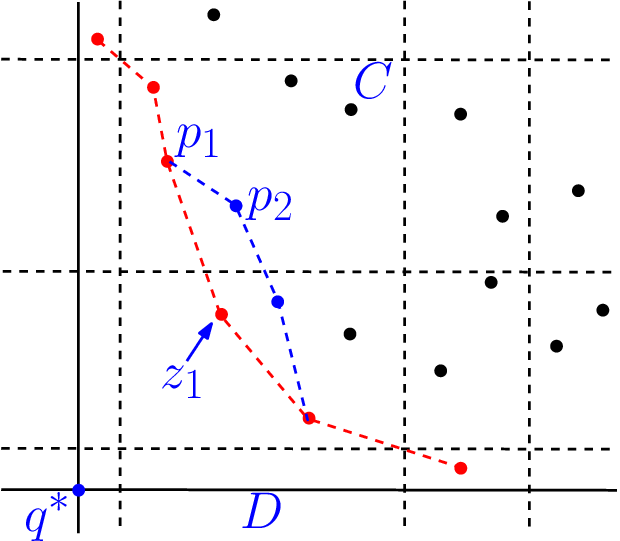}
\caption{\footnotesize The red points are in $\pi_1$ and the blue
points are in $\pi_2\setminus \pi_1$. The skyline-bottom point of the
cell $C$ is $p_1$ in $\pi_1$ but $p_2$ in $\pi_2$.
}
\label{fig:newskyline}
\end{center}
\end{minipage}
\vspace*{-0.15in}
\end{figure}

The above implies that to determine $\calC_2$, we need to do the
following. (1) Find all cells in $\calC'_2$, and as in Lemma
\ref{lem:70}, for each cell of $\calC'_2$, compute its skyline-left
point and skyline-bottom point as well as their generating segments.
(2) Determine whether $C_i$ is still in $\calC_2$, and if yes,
compute its new skyline-left
point and skyline-bottom point as well as their generating segments,
if any of them changes. (3) Compute the new skyline-bottom point (and
its generating segment) for $C_{i-1}$ if it changes.  (4)
Compute the new skyline-left point (and
its generating segment) for $C_{i+1}$ if it changes.

Let $\calD_P$ be the dynamic segment-dragging data structure in
Lemma \ref{lem:dynamic} we built on $P$.
Below, we give an algorithm that can determine $\calC_2$ in
$O((1+m'_2)(\log^2 n+\log m))$ time, and in particular, we need to find the
cells of $\calC'_2$. Intuitively, if $C_i\not\in \calC_2$, then
$\calC'_2$ consists of all cells
of $\calC_2$ between $C_{i-1}$ and $C_{i+1}$ in the
canonical order; otherwise, $\calC'_2$ consists of all cells
of $\calC_2$ between $C_{i-1}$ and $C_i$ and all cells between $C_i$
and $C_{i+1}$.
Our algorithm essentially follows the behavior of the algorithm in
Lemma \ref{lem:70}, but only focuses on the cells in $\calC'_2\cup
\{C_{i-1},C_i,C_{i+1}\}$. Recall that $C_i$ is the cell of $\calC_1$
that contains $z_1$.

First of all, we delete the point $z_1$ from the data structure $\calD_P$.
The point $z_1$ can be the skyline-left point of $C_i$,
or the skyline-bottom point of $C_i$, or both of them, or
neither of them.  Our algorithm
works differently for these cases, as follows. Recall that according
to our algorithm in Lemma \ref{lem:70}, if $z_1$ is either the
skyline-left point or the skyline-bottom point of $C_i$, then $z_1$
has a generating segment, denoted by $s(z_1)$. In other words, $z_1$ is
identified by a segment-dragging query on $s(z_1)$ in our algorithm in
Lemma \ref{lem:70}.

\begin{enumerate}

\item

If $z_1$ is neither the skyline-left point nor the skyline-bottom
point of $C_i$, then $C_i$ is still in $\calC_2$ and $\calC'_2=\emptyset$. In
fact, $\calC_2=\calC_1$. Further, the skyline-left and skyline-bottom
points of any cell of $\calC_1$ do not change. Hence, we are done for
this case.

\item

If $z_1$ is the skyline-left point but not the skyline-bottom point,
then according to our algorithm in Lemma \ref{lem:70}, the generating
segment $s(z_1)$ is either an
$s_2$-type or an $s_0$-type. Note that since $z_1$ is not the
skyline-bottom point of $C_i$, the skyline-bottom point of $C_i$ is
still in the skyline $\pi_2$, which implies that
$C_i$ is still in $\calC_2$ and no cell of $\calC'_2$ is between $C_i$ and
$C_{i+1}$ in the canonical order of $\calC_2$. In other words,
all cells of $\calC'_2$ are between $C_{i-1}$ and $C_i$ in the canonical
order of $\calC_2$. Denote by $D$ the column of $\calA$ that contains
$C_i$.

\begin{enumerate}
\item
If $s(z_1)$ is an $s_2$-type, then $C_i$ is not the topmost cell of
$\calC_1$ in the column $D$, which implies that
$C_{i-1}$ is in $D$.  According to the algorithm in Lemma
\ref{lem:70}, $s(z_1)$ is the left side of $C_i$ (i.e., $z_1$ is the
first point of $P$ hit by dragging $s(z_1)$ rightwards).
By using the data
structure $\calD_P$ (after deleting $z_1$), we do a segment-dragging
query by dragging $s(z_1)$ rightwards to find the first point of
$P\setminus \{z_1\}$ hit by $s(z_1)$, and we denote the point by $p$.
Then, $p$ is the new skyline-left point of $C_i$ (without considering
$z_1$). Note that $s(z_1)$ is still an $s_2$-type generating segment
for $p$.

Next, from $C_i$, we continue to find the cells of $\calC'_2$ in a
bottom-up manner in the same way as the algorithm in Lemma
\ref{lem:70} until we meet the cell $C_{i-1}$.
Note that it is possible that $\calC'_2=\emptyset$. Again,
it takes two segment-dragging queries (using $\calD_P$) on each cell of $\calC'_2$ to find
its skyline-left and skyline-bottom point as well as their generating
segments.
Also, the algorithm will find the new
skyline-bottom point of $C_{i-1}$ if it changes in $\pi_2$.
Recall that given any point $p$, we can determine the cell of $\calA$
that contains $p$ in $O(\log m)$ time (by binary search on the sorted
vertical lines of $\calA$ and on the sorted horizontal lines of
$\calA$). Therefore, in this case,
the total running time to determine $\calC_2$
is $O((1+m'_2)(\log^2 n+\log m))$ time.

\item
If $s(z_1)$ is an $s_0$-type, then $C_{i}$ is the topmost cell of
$\calC_1$ in the column $D$, which implies that $C_{i-1}$ is in a
column to the left of $D$. Denote by $D'$ the column of $\calA$
containing $C_{i-1}$ and let $p$
be the skyline-bottom point of $C_{i-1}$. According to the algorithm
in Lemma \ref{lem:70}, $s(z_1)$ is the vertical line segment on the right
side of $D'$ where the lower endpoint of $s(z_1)$ is on the
horizontal line $y=y(q^*)$ and the upper endpoint has the same
$y$-coordinate as $p$, and $z_1$ is the first point of $P$ hit by
dragging $s(z_1)$ rightwards.

By using the data structure $\calD_P$ (after deleting $z_1$),
we do a segment-dragging query by dragging $s(z_1)$ rightwards;
let $p'$ be the point returned by the query (i.e., $p'$ is the
first point of $P\setminus\{z_1\}$ hit by dragging $s(z_1)$ rightwards).
Note that $s(z_1)$ is still an $s_0$-type generating segment for $p'$.

\begin{enumerate}
\item

If $p'$ is in $C_i$, then $p'$ is the new skyline-left point of $C_i$,
and $C_i$ is still the topmost cell of $\calC_2$ in $D$, which implies
$\calC'_2=\emptyset$.

\item
\label{lab:200}

If $p'$ is not in $C_i$, then let $C'$ be the cell containing $p'$ and
$p'$ is the skyline-left point of $C'$.
Since $z_1$ is not the skyline-bottom point of $C_i$, the cell $C'$ is
still in the column $D$ and is higher than $C_i$ (see
Fig.~\ref{fig:cases}(a)). Then, from the cell
$C_i$ to $C'$,
we use the bottom-up procedure as in the algorithm in Lemma
\ref{lem:70} to find the cells of $\calC_2$ between $C_i$ and $C'$ in
the column $D$ and these cells (expect $C_i$) constitute the set $\calC'_2$.
Again, it takes two segment-dragging queries (using $\calD_P$)
for each cell of $\calC'_2$ to find
its skyline-left and skyline-bottom point as well as their generating
segments.
\end{enumerate}

The total running time is $O((1+m'_2)(\log^2 n+\log m))$ time.
\end{enumerate}

\item

If $z_1$ is the skyline-bottom point but not the skyline-left point,
then according to our algorithm in Lemma \ref{lem:70}, $s(z_1)$ is
either an
$s_1$-type or an $s_3$-type. Note that since $z_1$ is not the
skyline-left point of $C_i$, the skyline-left point of $C_i$ is still
in the skyline $\pi_2$, which implies that $C_i$ is still in $\calC_2$
and no cell of $\calC'_2$ is between $C_{i-1}$ and $C_i$ in the canonical
order of $\calC_2$. In other words, all cells of $\calC'_2$ are between
$C_{i}$ and $C_{i+1}$ in the canonical order of $\calC_2$.
Denote by $D$ the column of $\calA$ that contains $C_i$.

\begin{figure}[t]
\begin{minipage}[t]{\linewidth}
\centering
\begin{tabular}{ccc}
\includegraphics[totalheight=2.0in]{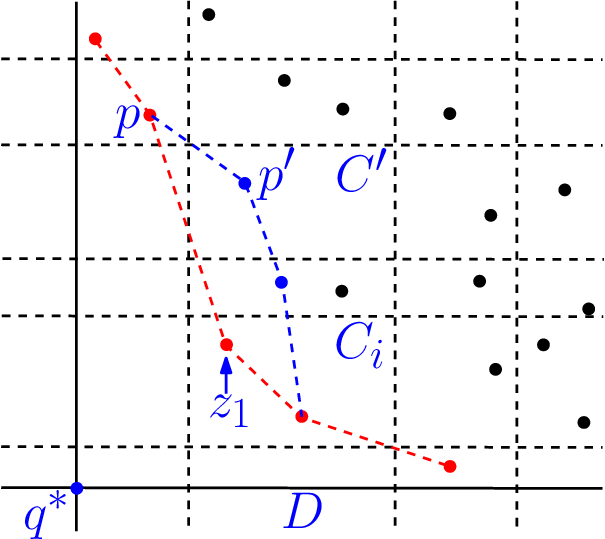}
& \hspace*{3mm}
&
\includegraphics[totalheight=2.0in]{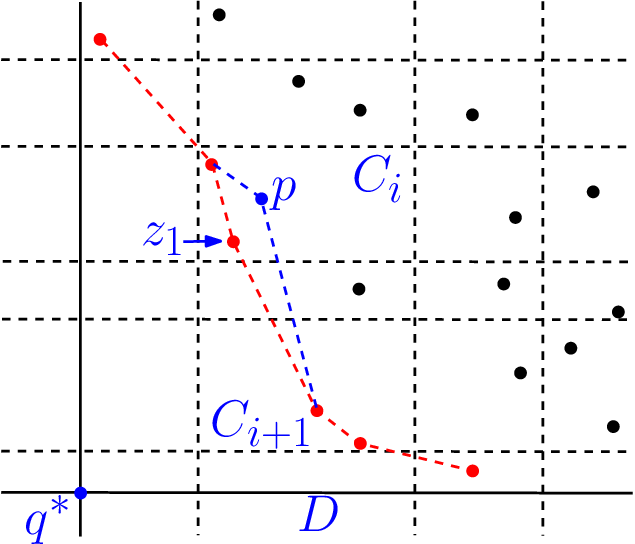} \\
(a) & & (b)
\end{tabular}
\caption{\footnotesize (a). Illustrating Case \ref{lab:200}:
If $C'\neq C_i$, then $C'$ is in $D$ and higher than $C_i$. (b). Illustrating Case \ref{lab:100}: $p$ is the
first point hit by dragging $s(z_1)$ upwards without considering $z_1$.
}
\label{fig:cases}
\end{minipage}
\vspace*{-0.15in}
\end{figure}

\begin{enumerate}
\item
\label{item:20}
If $s(z_1)$ is an $s_1$-type, then $C_i$ is the bottommost cell of
$\calC_1$ in $D$.
According to the algorithm in Lemma \ref{lem:70}, $s(z_1)$ is the
intersection of $D$ and the horizontal line
$y=y(q^*)$. By using the data structure $\calD_P$ (after deleting $z_1$),
we do a segment-dragging query by dragging $s(z_1)$
upwards and let $p$ be the point returned by the query. Then $p$ is
the new skyline-bottom point of $C_i$. Next, we find the cells in
$\calC'_2$.

Let $s'$ be the vertical segment on the right side of $D$
where the lower endpoint of $s'$ is on the
horizontal line $y=y(q^*)$ and the upper endpoint of $s'$ has the same
$y$-coordinate as $p$.
We do a segment-dragging query by
dragging $s'$ rightwards and let $p'$ be the point given by the query.
The segment $s'$ is the generating segment of $p'$, and in fact,
$s'$ is an $s_0$-type generating segment based on our definition in
the proof of Lemma \ref{lem:70}. Denote by $C(p')$ the
cell of $\calA$ that contains $p'$.
Let $D'$ be the column that contains $C_{i+1}$.

\begin{enumerate}
\item
If $C(p')$ is in $D'$, then there are further two cases.
If $C(p')$ is $C_{i+1}$, then $p'$ is the new skyline-left point of
$C_{i+1}$, and $\calC'_2=\emptyset$.
Otherwise, from $C_{i+1}$ to $C(p')$, we use
the same bottom-up procedure as in the algorithm in Lemma \ref{lem:70} to
find all cells of $\calC_2$ between $C_{i+1}$ and $C(p')$, and these
cells (expect $C_{i+1}$) constitute the set $\calC'_2$.

\item
If $C(p')$ is not in $D'$, then it must be in a column to the left of
$D'$. From the cell $C(p')$, we proceed in the same way as in the
algorithm in Lemma
\ref{lem:70} until the first time we find a cell in the column $D'$. Then, we
use the same algorithm as the above case where $C(p')$ is in $D'$.
\end{enumerate}

\item
\label{lab:100}

If $s(z_1)$ is an $s_3$-type, then $C_i$ is the not bottommost cell of
$\calC_1$ in $D$, which implies that $C_{i+1}$ is in $D$. We show
below that $\calC'_2=\emptyset$; further, we
will find a new skyline-bottom point in $C_i$ (without considering
$z_1$).

Based on our algorithm
in Lemma \ref{lem:70}, the generating segment $s(z_1)$ of $z_1$ is the
horizontal line segment on the top side of $C_{i+1}$ whose left
endpoint is the upper left vertex of $C_{i+1}$ and right endpoint has
the same $x$-coordinate as the skyline-left point of $C_{i+1}$.
By using the data
structure $\calD_P$ (after deleting $z_1$), we do a segment-dragging
query by dragging $s(z_1)$ upwards, and let $p$ be the point returned
by the query.

Note that $z_1$ is the lowest point of $P$ that will be hit by
dragging $s(z_1)$ upwards and $z_1$ is in $C_i$.
The point $p$ is the lowest point of $P\setminus\{z_1\}$ that will be hit by
dragging $s(z_1)$ upwards (see Fig.~\ref{fig:cases}(b)).
Clearly, $p$ cannot be any cell of $D$
lower than $C_i$. On the other hand, since $z_1$ is not the
skyline-left point of $C_i$, the skyline-left point of $C_i$ is still
in $C_i$. Note that when we drag $s(z_1)$ upwards, the skyline-left
point of $C_i$ will be hit by $s(z_1)$ (but not necessarily the first
point hit by $s(z_1)$), and this implies that the point $p$ must be in $C_i$.
In other words, $p$ is the skyline-bottom point of $C_i$ in the new
skyline $\pi_2$, and further $\calC'_2=\emptyset$.

\end{enumerate}

In any case above, the total running time is $O((1+m'_2)(\log^2 n+\log m))$ time.
\item

It remains to discuss the case where $z_1$ is both the skyline-left
point and the skyline-bottom point of $C_i$.
In this case, $z_1$ is the
first time identified as either the skyline-left point or the
skyline-bottom point. In general, unlike the second and the third
cases where the cells of $\calC'_2$ are either between $C_{i-1}$ and
$C_i$ or between $C_i$ and $C_{i+1}$ in the canonical order of
$\calC_2$, in this case the cells of $\calC'_2$ may lie both between $C_{i-1}$ and
$C_i$ and between $C_i$ and $C_{i+1}$. Hence, our algorithm may need to
search on both ``directions''.  In addition, in the previous three
cases, the cell $C_i$ must be in $\calC_2$; in this case, however, it is possible
that $C_i$ is not in $\calC_2$.

\begin{enumerate}
\item
If $z_1$ is the first time identified as the skyline-left point of
$C_i$, then $C_i$ must be the topmost cell of $\calC_1$ in $D$ where
$D$ is the column of $\calA$ that contains $C_i$, which implies that
its generating segment $s(z_1)$ must be an $s_0$-type.

Let $p$ be the skyline-bottom point of the cell $C_{i-1}$. Let $D'$ be
the column of $\calA$ that contains $C_{i-1}$.  According to the algorithm
in Lemma \ref{lem:70}, $s(z_1)$ is the vertical line segment on the right
side of $D'$ where the lower endpoint of $s(z_1)$ is on the
horizontal line $y=y(q^*)$ and the upper endpoint has the same
$y$-coordinate as $p$. By using the data structure $\calD_P$
(after deleting $z_1$), we do a segment-dragging query by
dragging $s(z_1)$ rightwards, and let $p'$ be the point given
by the query. Let $C(p')$ be the cell that contains $p'$.
Note that $p'$ is the skyline-left point of $C(p')$.
Let $D''$ be the column that contains the cell $C_{i+1}$. Note that it is possible that
$D''=D$.

\begin{enumerate}
\item
If $C(p')$ is also in $D''$, then there are further two cases.

If $C(p')=C_{i+1}$, then $\calC'_2=\emptyset$ and $C_i\not\in \calC_2$.

Otherwise, $C(p')$ must be higher than $C_{i+1}$ in $D''$.
Then, from the cell $C_{i+1}$ to $C(p')$,
we use the same bottom-up procedure as in the algorithm in Lemma
\ref{lem:70} to find all cells of $\calC_2$ between $C_{i+1}$ and
$C(p')$, and these cells (except $C_{i+1}$ and possibly $C_i$)
constitute the set $\calC'_2$. Note that the cell $C_i$ may or may not
be identified as in $\calC_2$ in the above procedure.

\item
\label{item:100}
If $C(p')$ is not in $D''$, then $C(p')$ must be in a column to the
left of $D''$. We proceed from $C(p')$ in the same way as in the
algorithm in Lemma \ref{lem:70} until the first time we find a cell in
$D''$. Then, we use the same algorithm as in the above case (i.) to
determine $\calC_2$.
\end{enumerate}

In any case, the total running time is $O((1+m'_2)(\log^2 n+\log m))$ time.

\item
If $z_1$ is the first time identified as the skyline-bottom point,
then its generating segment $s(z_1)$ can be either an $s_1$-type or an
$s_3$-type segment. Let $D$ be the column that contains $C_i$.
In this case, $C_i$ is not the topmost cell of $\calC_2$ in $D$ since
otherwise $z_1$ would be the first time identified as the skyline-left
point of $C_i$. This means that $C_{i-1}$ is also in $D$.

\begin{enumerate}

\item
If $s(z_1)$ is an $s_1$-type segment,
then $C_i$ must be the bottommost cell of $\calC_1$ in the column
$D$. According to the algorithm in Lemma \ref{lem:70}, $s(z_1)$ is the
intersection of $D$ and the horizontal line
$y=y(q^*)$. By using the data structure $\calD_P$ (after deleting $z_1$),
we do a segment-dragging query by dragging $s(z_1)$
upwards and let $p$ be the point returned by the query.
Let $C(p)$ be the cell that contains $p$. Clearly, $C(p)$ is in
$\calC_2$. Let $\calC'_{21}$ be the subset of cells in $\calC'_2$ that are
between $C_{i-1}$ and $C(p)$ in the canonical order of $\calC_2$, and
let $\calC'_{22}=\calC'_2\setminus \calC'_{21}$; in other words,
$\calC'_{22}$ is the subset of cells in $\calC'_2$ that are
between $C(p)$ and $C_{i+1}$ in the canonical order of $\calC_2$.
Below, we will find $\calC'_{21}$ and $\calC'_{22}$ separately, by searching
from $C(p)$ towards two ``directions'': one towards $C_{i-1}$ and the
other towards $C_{i+1}$.

Since $C_{i-1}$ is also in $D$, from $C(p)$ to $C_{i-1}$, we use the
bottom-up procedure to find the cells of $\calC_2$ between $C(p)$
and $C_{i-1}$, and these cells (except $C_{i-1}$ and possibly $C_i$)
constitute the set $\calC'_{21}$.
Note that the cell $C_i$ may also be identified in $\calC_2$.

%If $C_{i-1}$ is not in $D$, then $C_{i-1}$ must be in a column to the
%left of $D$. We use the same algorithm as in the above Case
%\ref{item:100} to find the cells of $\calC_2$, and these cells
%except $C_{i-1}$ and possibly $C_i$ consititue the set $S_1$.

Next, we find the set $\calC'_{22}$, which can be done by the same
algorithm as in Case \ref{item:20}. We omit the details.

\item
If $s(z_1)$ is an $s_3$-type, then $C_i$ is the not bottommost cell of
$\calC_1$ in $D$, which implies that $C_{i+1}$ is in $D$.

Based on our algorithm
in Lemma \ref{lem:70}, the generating segment $s(z_1)$ of $z_1$ is the
horizontal line segment on the top side of $C_{i+1}$ whose left
endpoint is the upper left vertex of $C_{i+1}$ and right endpoint has
the same $x$-coordinate as the skyline-left point of $C_{i+1}$.
By using the data
structure $\calD_P$ (after deleting $z_1$), we do a segment-dragging
query by dragging $s(z_1)$ upwards, and let $p$ be the point returned
by the query. Let $C(p)$ be the cell that contains $p$.

Note that $z_1$ is the lowest point of $P$ that will be hit by
dragging $s(z_1)$ upwards and $z_1$ is in $C_i$.
The point $p$ is the lowest point of $P\setminus\{z_1\}$ that will be hit by
dragging $s(z_1)$ upwards. Since $C_{i-1}$ is also in $D$, $C_{i-1}$
is higher than $C_i$. Hence, the cell $C(p)$ is one of the cells of $D$ between
(and including) $C_{i-1}$ and $C_i$ (this is because
the vertices of $\pi_1$ in $C_{i-1}$ are
all to the left of the right endpoint of $s(z_1)$ and to the right of
the left endpoint of $s(z_1)$). Hence, from $C(p)$ to $C_{i-1}$,
we can use the bottom-up
procedure as before to find the cells of $\calC_2$, these cells (except
$C_{i-1}$ and possibly $C_i$) constitute the set $\calC'_2$.
\end{enumerate}

In any case, the total running time is $O((1+m'_2)(\log^2 n+\log m))$ time.
\end{enumerate}
\end{enumerate}

In summary, we can determine the set $\calC_2$ in $O((1+m'_2)(\log^2 n+\log m))$
time. More specifically, for each cell in $\calC'_2$, we have computed its skyline-left
point and its skyline-bottom point as well as their generating
segments. We have also determined whether $C_i$ is in $\calC_2$, and
if yes, its new skyline-bottom point and skyline-left point are
computed if any of them changes.
The new skyline-bottom point of $C_{i-1}$ has
been found if it changes, and the new skyline-left point of $C_{i+1}$ has
been found if it changes. In addition, in the above algorithm,
we can also order all cells of $\calC'_2$ (with $C_i$ if $C_i\in \calC_2$)
from northwest to southeast with the same running time, and therefore, along with the ordered
cells from $C_1$ to $C_{i-1}$ and the ordered cells from $C_{i+1}$ to
$C_{m_1}$, we have obtained a canonical order for~$\calC_2$.
\end{proof}

By Lemma \ref{lem:100}, we can determine the set $\calC_2$, and in
particular, we have the set $\calC'_2$ explicitly, and we know whether
the cell $C(z_1) \in \calC_2$. Similarly to Lemma \ref{lem:65}, the
second $\ANN$ $z_2$ is in one of the cells of $\calC_2$.
Denote by $P_1=P\setminus\{z_1\}$.

To find $z_2$, as in the case for finding $z_1$, a straightforward
approach is to compute the $\ANN$ of $Q$ in
$P_1\cap C$ for each $C\in \calC_2$, and then among the
$|\calC_2|$ candidate points, report the one with the smallest aggregate distance to $Q$
as $z_2$. This approach will lead to an
$O(km)$ time query algorithm for finding $S_k(P^1)$.
Below, we present a better method.

Note that when computing $z_1$, we have computed the
$\ANN$ $S_1(P\cap C)$ for each $C \in \calC_1$. Also, for each cell
$C\in \calC_1$, if $C\neq C(z_1)$, then $C\in \calC_2$ and $P\cap C=P_1\cap C$.
Therefore, if we maintain the $\ANN$s for all cells of
$\calC_{1}\setminus{C(z_1)}$, we do not have to compute them again.
In other words, when computing $z_2$, we only need to compute
the $\ANN$s in the cells of $\calC'_2$. In addition, if $C(z_1) \in \calC_2$,
we will use a special approach to compute $S_1(P_1\cap C(z_1))$.
To maintain the $\ANN$s in the involved cells mentioned above,
we use a min-heap $H$, as follows.

When searching $z_1$, for each $C\in \calC_1$,
after the $\ANN$ $S_1(P\cap C)$ is computed, we insert it into $H$ with
its aggregate distance to $Q$ as the ``key''. After the $\ANN$s for
all cells of $\calC_1$ are computed and inserted into $H$, the point in $H$
with the smallest key is $z_1$. Note that $H$ has $m_1 = |\calC_1|$ points.
To compute the second $\ANN$ $z_2$, we
first determine $\calC'_2$ by Lemma \ref{lem:100}.
By the ``Extract-Min'' operation
of min-heaps \cite{ref:CLRS09}, we remove $z_1$ from $H$.
We compute the $\ANN$s of the cells in $\calC_2'$ and insert them into
$H$. If $C(z_1)\not\in \calC_2$, then the point of $H$ with
the smallest key is $z_2$. Otherwise, we use the following special
approach to determine $S_1(P_1\cap C(z_1))$.

One tempting approach is to have a dynamic version of the data
structure in Lemma \ref{lem:80} to support point deletions from $P$.
Unfortunately, due to the ``static'' nature of compact interval
trees, it is not clear to us how to design such a dynamic data
structure without deteriorating the performance.  Instead, we present
another method to ``mimic'' point deletions, as follows.

We divide the cell $C(z_1)$ into two sub-cells $C_1(z_1)$ and $C_2(z_1)$
using the horizontal line through $z_1$. Hence, $z_1$ is on the common
edge of the two sub-cells.  Note that due to our general
position assumption, no point of $P$ is on the boundary of $C(z_1)$.
Hence, no point of $P_1=P\setminus\{z_1\}$ is on the boundary of
$C_1(z_1)$ (or $C_2(z_1)$). Below, we use $C_1(z_1)$ (resp.,
$C_2(z_1)$) to refer to only its interior.
Instead of computing the $\ANN$ $S_1(P_1\cap C(z_1))$ and insert it into $H$,
we compute the $\ANN$s $S_1(P\cap C_1(z_1))$ and $S_1(P\cap C_2(z_1))$
and insert them into $H$; note that one of them is
$S_1(P_1\cap C(z_1))$. The reason we divide  $C(z_1)$ into two
sub-cells as above is that we can now
simply use the data structure in Lemma \ref{lem:80} to compute
$S_1(P\cap C_1(z_1))$ and $S_1(P\cap C_2(z_1))$; in other words, $z_1$
appears to be ``deleted'' from the data structure of
Lemma \ref{lem:80}. Clearly, now, the point of $H$ with smallest key is $z_2$.

To analyze the running time for computing $z_2$, $\calC_2$ can be
determined in $O((1+m'_2)(\log^2 n+\log m))$ time, after which, we
compute the $\ANN$s for the cells of $\calC'_2$ and possibly for the two
sub-cells of $C(z_1)$ in $O((2+m'_2)\log^2 n)$ time by Lemma \ref{lem:80}. Then,
one ``Extract-Min'' operation and at most $m_2'+2$ insertions
on $H$ together take $O((m'_2+3)\log(|H|))$ time; note that $|H|\leq m_1+m_2'+2$
(here ``2'' corresponds to the number of possible sub-cells).

It should be noted that we need to explicitly
maintain the two sub-cells $C_1(z_1)$ and $C_2(z_1)$
because later they may be further divided into smaller sub-cells
(e.g., if $z_2\in C_1(z_1)$ and $C(z_1)\in \calC_3$, then $C_1(z_1)$
will be divided for computing $z_3$).
Also note that these sub-cells are only maintained for computing
$\ANN$s and they will not be considered when we determine the sets
$\calC_i$'s (in Lemma \ref{lem:100}).
After $z_2$ is found, we proceed to search the third $\ANN$ $z_3$ similarly.

In general, suppose we have computed $\calC_i$
and $z_i$, and we are about to find $z_{i+1}$.
We first determine $\calC_{i+1}$ by computing
$\calC_{i+1}'$ and determining whether $C(z_i)\in \calC_{i+1}$, where $C(z_i)$ is the cell of $\calC_{i}$ that
contains $z_i$; this can be done in $O((1+m'_{i+1})(\log^2 n+\log m))$ time
similarly as in Lemma \ref{lem:100}.

Note that for any cell $C\in \calC'_{i+1}$, it never
appears in $\calC_j$ for any $1\leq j\leq i$.
Next, we determine the $\ANN$s in the cells of $\calC'_{i+1}$ by Lemma
\ref{lem:80} and insert them into the heap $H$. We also need to remove $z_i$
from $H$. If $C(z_i)\notin \calC_{i+1}$, then the point of $H$ with
smallest key is $z_{i+1}$. Otherwise, as before, we
divide $C(z_i)$ into two sub-cells and compute their
$\ANN$s and insert them into $H$. Note that $C(z_i)$ may have already
been divided into many sub-cells before. If so, they are explicitly
maintained, and we can find the sub-cell that contains
$z_i$ in $O(\log k)$ time by binary search
since $C(z_i)$ has at most $k-1$ sub-cells ordered by $y$ values. Then, we divide the
sub-cell into two smaller sub-cells
by the horizontal line through $z_i$ and compute the
$\ANN$s in the two smaller sub-cells by Lemma \ref{lem:80} and insert
them into $H$. Now, the point of $H$ with smallest key is $z_{i+1}$.

To analyze the running time for computing $z_{i+1}$, $\calC_{i+1}$ can be
determined in $O((1+m'_{i+1})(\log^2 n+\log m))$ time. The time for
computing the $\ANN$s for the cells in $\calC'_{i+1}$ and possibly two
sub-cells is bounded by $O((2+m'_{i+1})\log^2 n)$. There are
$O(2+m'_{i+1})$ insertions and one ``Extract-Min'' operation on $H$,
which together take $O((m'_{i+1}+3)\log(|H|))$ time. Note
that $|H|\leq m_i+m'_{i+1}+2$.

We repeat the above procedure until $z_k$ is found.
We have the following lemma (a crucial
observation is that $m_1+\sum_{i=2}^km'_i=O(m+k)$).

\begin{lemma}\label{lem:110}
The overall running time of our query algorithm for finding
$S_k(P^1)=\{z_1,z_2,\ldots,z_k\}$ is $O(m\log m+(k+m)\log^2 n)$.
\end{lemma}
\begin{proof}
\def\total{\lambda}
Let $\total=m_1+\sum_{i=2}^km'_i$ denote the total number of cells in $\calC_1\cup
\bigcup_{i=2}^k\calC'_i$.

By Lemma \ref{lem:70}, we compute $\calC_1$ in $O(m\log n+m\log m)$
time.  By Lemma \ref{lem:100}, the
total time for finding all cells of $\bigcup_{i=2}^k\calC'_i$ is
$O((k+\total)(\log^2n +\log m))$.

In the entire algorithm, the total number of operations for
finding the $\ANN$s in the cells of $\calA$ (not including the sub-cells)
is $O(\total)$ because the above cells are those in $\calC_1\cup
\bigcup_{i=2}^k\calC'_i$.
After finding $z_i$ for each $1\leq i\leq k$, we have
at most two more sub-cells, and thus the total number of operations for
finding the $\ANN$s in the sub-cells is $O(k)$.
Hence, by Lemma \ref{lem:80}, the total time for finding the $\ANN$s in the cells and sub-cells
is $O((\total+k)\log^2 n)$. Also, we only need to explicitly maintain at
most $O(k)$ sub-cells in the entire algorithm.

Similarly, the total number of operations on the heap $H$ is
$O(\total+k)$, and the size of $H$ in the entire algorithm is
always bounded by $O(\total+k)$. Hence, the total operations on $H$ take
$O((\total+k)\log (\total+k))$ time.

In summary, the overall running time
%of the entire algorithm for finding $S_k(P^1)$
is $O((k+\total)(\log^2n +\log m+\log (\total+k)))$.
To prove the lemma, we prove an important {\em claim}: $\total=O(m+k)$.

The proof for the claim is based on the fact that $|\calC_i|=O(m)$
for each $1\leq i\leq k$, since each skyline $\pi_i$ intersects
$O(m)$ cells. In particular, $|\calC_k|=O(m)$. For each
$1\leq i\leq k-1$, all the cells of
$\calC_i$ except $C(z_i)$ are in $\calC_{i+1}$ and the
cell $C(z_i)$ may or may not be in $\calC_{i+1}$. Hence,
$|\calC_{i+1}|\geq |\calC_i|-1+|\calC'_{i+1}|$, i.e., $m_{i+1}\geq
m_i-1+m'_{i+1}$. Therefore,
$m_k\geq m_1+\sum_{i=2}^km'_i-(k-1)$.
Due to $m_k=O(m)$, we have $\total=m_1+\sum_{i=2}^km'_i\leq m_k+k-1=O(m+k)$.
The above claim thus follows.

Due to the above claim, the overall running time for finding $S_k(P^1)$ is
$O((k+m)(\log^2 n+\log (k+m)))$, which is
$O(m\log m+(k+m)\log^2 n)$ (to see this, note that if $k>m $, then
since $n\geq k$, $(k+m)\log (k+m)=O((k+m)\log^2 n)$ holds).
\end{proof}

Note that after obtaining $S_k(P^1)$, we also need to insert the points of $S_k(P^1)$ back to the data structure in Lemma \ref{lem:dynamic} for answering other top-$k$ $\ANN$ queries in future.

\subsection{Wrapping Things Up}

%A summary for the top-$k$ $\ANN$ searching is given in the proof of Theorem \ref{theo:2d}.
%Note that we also need to insert the points of $S_k(P)$ back to the data structure in Lemma \ref{lem:dynamic} for answering other top-$k$ $\ANN$ queries in future.
We summarize our methods for the top-$k$ $\ANN$ queries in the $L_1$ metric.

Our preprocessing on $P$ includes the following steps. (1)
Sort all points in $P$ by their $x$-coordinates and $y$-coordinates,
respectively. (2) Build the dynamic segment-dragging query data
structure in Lemma \ref{lem:dynamic}
on $P$. (3) Construct the data structure in Lemma \ref{lem:80}.
The total time and space are dominated by Step
(3), regardless of which data structure of Lemma \ref{lem:80} is used.
%, i.e., $O(n\log n\log\log n)$ time and $O(n\log n\log\log n)$ space.

Given any query set $Q$ and any $k$, we compute
$S_k(P)$ in the following steps. (1) Sort all points in
$Q$ by their their $x$-coordinates and $y$-coordinates,
respectively. (2) Process $Q$ as in Lemma \ref{lem:50}.
(3) Compute a global minimum point $q^*$.
(4) Divide the plane into four quadrants with respect to $q^*$. In
each quadrant $R$, we find the top-$k$ $\ANN$s of $Q$ in $P\cap R$ as follows.
Suppose $R$ is the first quadrant. (4.1) Find the set
$\calC_1$ by Lemma \ref{lem:70}, and for each cell $C\in \calC_1$,
find the $\ANN$ of $Q$ in $P\cap C$ by Lemma \ref{lem:80} and insert
the point into a min-heap $H$; the point of $H$ with smallest
aggregate distance to $Q$ is the $\ANN$ of $Q$ in $P\cap R$.
(4.2) Based on $\calC_1$ and $z_1$, determine $\calC_2$ and
find $z_2$. (4.3) The above procedure continues until we
find $z_k$.
(5) Among the found $4k$ points from all four quadrants of $q^*$
(their aggregate distances to $Q$ have also been computed), we report
the $k$ points with smallest aggregate distances to $Q$ as
$S_k(P)$. (6) Insert the above $4k$ points back to the data structure
in Lemma \ref{lem:dynamic} (for answering other top-$k$ $\ANN$ queries in
future).

For the running time of the query algorithm, the first three steps can
be done in $O(m\log m)$ time; Step (4) can be done in
$O(m\log m+(k+m)\log^2 n)$ time. Step (5) takes $O(k)$
time. Step (6) needs $O(k\log^2 n)$ time. Hence, the total query time is bounded by $O(m\log m+(k+m)\log^2 n)$. If we use the other two data structures in Lemma \ref{lem:80}, then we have the query times of $O(m\log m+(k+m)\log^2 n\log\log n)$ and $O(m\log m+(k+m)\log^2 n\log^*n)$, respectively.

\begin{theorem}\label{theo:2d}
Given a set $P$ of $n$ points in the plane,
a data structure of $O(n\log n\log\log n)$ size can be built in
$O(n\log n\log\log n)$ time, such that for any weighted set $Q$ and any $k$, the
top-$k$ $\ANN$s can be found in $O(m\log m + (k+m)\log^2 n)$ time.
With trade-off between preprocessing and query time, we also build two other data structures: the first one has $O(n\log n)$ preprocessing time and space with $O(m\log m+(k+m)\log^2 n\log\log n)$ query time; the second one has $O(n\log n\log^* n)$ preprocessing time and space with $O(m\log m+(k+m)\log^2 n\log^* n)$ query time.
\end{theorem}

\section{Top-$k$ Aggregate Farthest Neighbor ($\AFN$) Searching in the $L_1$ Metric}
\label{sec:AFN}

Our techniques can be extended to solve the top-$k$ $\AFN$ searching, with the
same time bounds as in Theorems \ref{theo:1d} and \ref{theo:2d}.

For the 1-D case, given a query set $Q$ and $k$, recall that for
computing the top-$k$ $\ANN$s, we first compute the global minimum point
$q^*$ and then search in $P$ beginning from $q^*$ simultaneously
towards left and
right. To compute the top-$k$ $\AFN$s, due to Lemma \ref{lem:1dmonotone}, we search in
$P$ simultaneously beginning from the leftmost and rightmost points of $P$ and
towards the middle (e.g., either the leftmost or the rightmost point
of $P$ is the top-1 $\AFN$ of $Q$ by Lemma \ref{lem:1dmonotone}). The rest of the algorithm is similar
as that for Theorem \ref{theo:1d} and we omit the details. Hence, we
can obtain the following result.

\begin{theorem}\label{theo:1dAFN}
Given a set $P$ of $n$ points on the real line $L$, with
$O(n\log n)$ preprocessing time and $O(n)$ space, the top-$k$ $\AFN$s
can be found in $O(\min\{k,\log m\}\cdot m+k+\log n)$ time
for any query set $Q$ and any $k$; if the points of $Q$
are given sorted on $L$, then the query time is $O(k+ m+ \log
n)$.
\end{theorem}

For the 2-D case, consider any query set $Q$ and $k$. As in the $\ANN$ case,
we first compute
a global minimum point $q^*$ and then compute the top-$k$ $\AFN$s
in each quadrant of $q^*$. Suppose $R$ is the first quadrant with
respect to $q^*$. We find the top-$k$ $\AFN$s in $R$ as follows.
Recall that in the $\ANN$ case we search the top-$k$ $\ANN$s in a direction
from $q^*$ towards northeast. In the $\AFN$ case, due to Lemma
\ref{lem:2dmonotone}, we search the top-$k$ $\AFN$s along
the opposite direction, i.e., from northeast towards $q^*$. Specifically, here we
re-define the ``dominate'' relationship in the opposite way as before: a
point $p_1$ {\em dominates} $p_2$ if and only if $x(p_1)\geq x(p_2)$
and $y(p_1)\geq y(p_2)$. A
point $p$ in $P\cap R$ is called a {\em maximal point} if no other point
in $P\cap R$ dominates $p$. Similarly, we re-define the {\em skyline}
as the path connecting all maximal points of $P\cap R$. According to Lemma
\ref{lem:2dmonotone}, the $\AFN$ of $Q$ in $P\cap R$ must be in the skyline.
Then, we can use a similar algorithm as in the $\ANN$ case to compute all
top-$k$ $\AFN$s. More
specifically, we first compute the $\AFN$ $p$ on the skyline and then search
the second $\AFN$ on the next skyline (without considering $p$); we continue this procedure until
we find the $k$-th $\AFN$. The algorithm is similar (or symmetric) as the
$\ANN$ case and we omit the details. Hence, we can obtain the following
result.

\begin{theorem}\label{theo:2dAFN}
Given a set $P$ of $n$ points in the plane,
a data structure of $O(n\log n\log\log n)$ size can be built in
$O(n\log n\log\log n)$ time, such
that for any weighted set $Q$ and integer $k$, the
top-$k$ $\AFN$s can be found in $O(m\log m + (k+m)\log^2 n)$ time.
With trade-off between preprocessing and query time, we also build two other data structures: the first one has $O(n\log n)$ preprocessing time and space with $O(m\log m+(k+m)\log^2 n\log\log n)$ query time; the second one has $O(n\log n\log^* n)$ preprocessing time and space with $O(m\log m+(k+m)\log^2 n\log^* n)$ query time.
\end{theorem}

\section{Conclusions}
\label{sec:conclusion}
We presented efficient methods for the top-$k$ aggregate nearest and farthest neighbor searching in the plane under the $L_1$ metric. Our results are the first-known solutions for the general top-$k$
queries on the weighted query points. Even for the special case where
$k=1$ or the unweighted query points, our results are generally better
than the previous work. While it would be interesting to investigate whether any further improvements are possible, another open problem is whether and how the
techniques proposed in this paper can be extended to
higher dimensional spaces.

\vspace*{0.1in}
\noindent\textbf{Acknowledgments.} The authors would like to thank Pankaj K. Agarwal
 for helpful discussions in early phases of this work.

%\begin{figure}[t]
%\begin{minipage}[t]{0.49\linewidth}
%\begin{center}
%\includegraphics[totalheight=0.8in]{window.eps}
%\caption{\footnotesize Illustrating a window $\overline{vw}$ of $p$.}
%\label{fig:window}
%\end{center}
%\end{minipage}
%\hspace{0.02in}
%\begin{minipage}[t]{0.49\linewidth}
%\begin{center}
%\includegraphics[totalheight=0.8in]{subproblem.eps}
%\caption{\footnotesize Illustrating the two critical constraints
%$\overline{vp_v}$ and $\overline{up_u}$ defined by
%the two mutually visible vertices $u$ and $v$.}
%\label{fig:subproblem}
%\end{center}
%\end{minipage}
%\vspace*{-0.25in}
%\end{figure}

%==========================end of document===========================

%\bibliographystyle{splncs03}
%\bibliography{reference}

%\vspace*{-0.10in}
%\footnotesize
%\baselineskip=11.0pt
\bibliographystyle{plain}
%\bibliography{refLib$\ANN$-original}
%\bibliography{refLib$\ANN$}
\bibliography{reference}

\begin{thebibliography}{10}

\bibitem{ref:AgarwalNe13}
P.K. Agarwal, B.~Aronov, S.~Har-Peled, J.M. Phillips, K.~Yi, and W.~Zhang.
\newblock Nearest neighbor searching under uncertainty {II}.
\newblock In {\em Proc. of the 32nd Symposium on Principles of Database Systems
  (PODS)}, pages 115--126, 2013.

\bibitem{ref:AgarwalNe12}
P.K. Agarwal, A.~Efrat, S.~Sankararaman, and W.~Zhang.
\newblock Nearest-neighbor searching under uncertainty.
\newblock In {\em Proc. of the 31st Symposium on Principles of Database Systems
  (PODS)}, pages 225--236, 2012.

\bibitem{ref:AhnGr13}
H.-K. Ahn, S.W. Bae, and W.~Son.
\newblock Group nearest neighbor queries in the {$L_1$} plane.
\newblock In {\em Proc. of the 10th Annual Conference on Theory and
  Applications of Models of Computation (TAMC)}, pages 52--61, 2013.

\bibitem{ref:AurenhammerVo00}
F.~Aurenhammer and R.~Klein.
\newblock {\em {\em Voronoi Diagram, in} Handbook of Computational Geometry,
  {\em J.-R Sack and J. Urrutia (eds.)}}, chapter~8, pages 201--290.
\newblock Elsevier, Amsterdam, the Netherlands, 2000.

\bibitem{ref:BeskalesEf08}
G.~Beskales, M.A. Soliman, and I.F. IIyas.
\newblock Efficient search for the top-{$k$} probable nearest neighbors in
  uncertain databases.
\newblock In {\em Proc. of the VLDB Endowment}, pages 326--339, 2008.

\bibitem{ref:ChazelleAn88}
B.~Chazelle.
\newblock An algorithm for segment-dragging and its implementation.
\newblock {\em Algorithmica}, 3(1--4):205--221, 1988.

\bibitem{ref:ChengPr08}
R.~Cheng, J.~Chen, M.~Mokbel, and C.-Y. Chow.
\newblock Probabilistic verifiers: {Evaluating} constrained nearest-neighbor
  queries over uncertain data.
\newblock In {\em Proc. of the 24th International Conference on Data
  Engineering (ICDE)}, pages 973--982, 2008.

\bibitem{ref:ChengUV10}
R.~Cheng, X.~Xie, M.L. Yiu, J.~Chen, and L.~Sun.
\newblock {UV-diagram: a Voronoi diagram for uncertain data}.
\newblock In {\em Proc. of the 26th International Conference on Data
  Engineering (ICDE)}, pages 796--807, 2010.

\bibitem{ref:ClarksonNe06}
K.L. Clarkson.
\newblock {\em {\em Nearest-Neighbor Searching and Metric Space Dimensions, in}
  Nearest-Neighbor Methods in Learning and Vision: Theory and Practice, {\em G.
  Shakhnarovich, T. Darrell and P. Indyk, (eds.)}}, chapter~2.
\newblock MIT Press, Cambridge, MA, 2006.

\bibitem{ref:CLRS09}
T.~Cormen, C.~Leiserson, R.~Rivest, and C.~Stein.
\newblock {\em Introduction to Algorithms}.
\newblock MIT Press, 3nd edition, 2009.

\bibitem{ref:deBergCo08}
M.~de~Berg, O.~Cheong, M.~van Kreveld, and M.~Overmars.
\newblock {\em Computational Geometry --- Algorithms and Applications}.
\newblock Springer-Verlag, Berlin, 3rd edition, 2008.

\bibitem{ref:DurierGe85}
R.~Durier and C.~Michelot.
\newblock Geometrical properties of the {Fermat-Weber} problem.
\newblock {\em European Journal of Operational Research}, 20:332--343, 1985.

\bibitem{ref:GaoAg11}
Y.~Gao, L.~Shou, K.~Chen, and G.~Chen.
\newblock Aggregate farthest-neighbor queries over spatial data.
\newblock In {\em Proc. of the 16th International Conference on Database
  Systems for Advanced Applications: Part II (DASFAA)}, pages 149--163, 2011.

\bibitem{ref:GuibasCo91}
L.~Guibas, J.~Hershberger, and J.~Snoeyink.
\newblock Compact interval trees: A data structure for convex hulls.
\newblock {\em International Journal of Computational Geometry and
  Applications}, 1(1):1--22, 1991.

\bibitem{ref:GuttmanR84}
A.~Guttman.
\newblock R-trees: a dynamic index structure for spatial searching.
\newblock In {\em Proc. of the ACM SIGMOD International Conference on
  Management of Data}, pages 47--57, 1984.

\bibitem{ref:LiGr11}
F.~Li, B.~Yao, and P.~Kumar.
\newblock Group enclosing queries.
\newblock {\em IEEE Transactions on Knowledge and Data Engineering},
  23:1526--1540, 2011.

\bibitem{ref:LiTw05}
H.~Li, H.~Lu, B.~Huang, and Z.~Huang.
\newblock Two ellipse-based pruning methods for group nearest neighbor queries.
\newblock In {\em Proc. of the 13th Annual ACM International Workshop on
  Geographic Information Systems}, pages 192--199, 2005.

\bibitem{ref:LiFl11}
Y.~Li, F.~Li, K.~Yi, B.~Yao, and M.~Wang.
\newblock Flexible aggregate similarity search.
\newblock In {\em Proc. of the ACM SIGMOD International Conference on
  Management of Data}, pages 1009--1020, 2011.

\bibitem{ref:LianPr08}
X.~Lian and L.~Chen.
\newblock Probabilistic group nearest neighbor queries in uncertain databases.
\newblock {\em IEEE Transactions on Knowledge and Data Engineering},
  20:809--824, 2008.

\bibitem{ref:LjosaAP07}
V.~Ljosa and A.~K. Singh.
\newblock {APLA: Indexing} arbitrary probability distributions.
\newblock In {\em Proc. of the 23rd IEEE International Conference on Data
  Engineering (ICDE)}, pages 946--955, 2007.

\bibitem{ref:LuoEf07}
Y.~Luo, H.~Chen, K.~Furuse, and N.~Ohbo.
\newblock Efficient methods in finding aggregate nearest neighbor by
  projection-based filtering.
\newblock In {\em Proc. of the 12nd International Conference on Computational
  Science and its Applications}, pages 821--833, 2007.

\bibitem{ref:PapadiasGr04}
D.~Papadias, Q.~Shen, Y.~Tao, and K.~Mouratidis.
\newblock Group nearest neighbor queries.
\newblock In {\em Proc. of the 20th International Conference on Data
  Engineering (ICDE)}, pages 301--312, 2004.

\bibitem{ref:PapadiasAg05}
D.~Papadias, Y.~Tao, K.~Mouratidis, and C.K. Hui.
\newblock Aggregate nearest neighbor queries in spatial databases.
\newblock {\em ACM Transactions on Database Systems}, 30:529--576, 2005.

\bibitem{ref:SharifzadehVo10}
M.~Sharifzadeh and C.~Shahabi.
\newblock {VoR-Tree: R-trees with Voronoi} diagrams for efficient processing of
  spatial nearest neighbor queries.
\newblock In {\em Proc. of the VLDB Endowment}, pages 1231--1242, 2010.

\bibitem{ref:WangAg13}
H.~Wang.
\newblock Aggregate-max nearest neighbor searching in the plane.
\newblock {\em arXiv:1309.1807v1}, 2013.
\newblock A preliminary version appeared in {\em Proc. of CCCG 2013}.

\bibitem{ref:YiuAg05}
M.L. Yiu, N.~Mamoulis, and D.~Papadias.
\newblock Aggregate nearest neighbor queries in road networks.
\newblock {\em IEEE Transactions on Knowledge and Data Engineering},
  17:820--833, 2005.

\bibitem{ref:YuenSu10}
S.M. Yuen, Y.~Tao, X.~Xiao, J.~Pei, and D.~Zhang.
\newblock Superseding nearest neighbor search on uncertain spatial databases.
\newblock {\em IEEE Transactions on Knowledge and Data Engineering},
  22:1041--1055, 2010.

\end{thebibliography}

%add appendix below
%\newpage
%\normalsize
%\appendix
%
%\section*{APPENDIX}

%\vspace{0.2in}
%\noindent
%{\bf Lemma \ref{lem:20}.}
%{\em The size of the set $C(s)$ is $O(k)$.
%}
%\vspace{0.08in}

\end{document}